\documentclass[11pt,draftcls,onecolumn]{IEEEtran}

\usepackage{tcolorbox}
\usepackage{cite}
\usepackage{float}		
\usepackage{amsfonts,epsfig,amsmath,latexsym,amssymb,amscd,multirow,graphicx,geometry,lscape,amsmath,amssymb,bm,pifont,graphicx,amssymb,amsmath}
\usepackage[autostyle]{csquotes} 
\usepackage{cancel,algorithm,algorithmic,algorithm,verbatim,array,multicol,palatino,amsfonts}

\usepackage{tabu,makecell}

\usepackage{enumitem}
\usepackage{graphicx}
\usepackage{lineno}
\usepackage{caption}
\usepackage{subcaption}
\usepackage[normalem]{ulem}
\usepackage{graphicx}
\usepackage{bbm}
\usepackage{dsfont}
\usepackage{amsthm}
\usepackage{amsmath}
\usepackage{arydshln}
\usepackage{mathrsfs}
\usepackage[framemethod=TikZ]{mdframed}
\usepackage{lipsum}
\usepackage{framed}

\newcommand{\x}{\mathbf{x}}

\newcommand{\g}{\mathbf{g}}

\renewcommand{\v}{\mathbf{v}}

\newcommand{\Z}{\mathbf{Z}}

\newcommand{\f}{\mathbf{f}}

\newcommand{\Y}{\mathbf{Y}}
\newcommand{\B}{\mathbf{B}}

\newcommand{\K}{\mathbf{K}}
\newcommand{\V}{\mathbf{V}}
\newcommand{\G}{\mathbf{G}}

\newcommand{\z}{\mathbf{z}}

\newcommand{\exE}{\mathbb{E}}

\newcommand{\C}{\mathcal{C}}
\renewcommand{\K}{\mathcal{K}}

\newcommand{\PSI}{\bm{\Psi}}
\newcommand{\THETA}{\bm{\theta}}

\newcommand{\xN}{\x_{\text{\tiny{1:N} } } }
\newcommand{\xNA}{\x_{\text{\tiny{1:N}}_\text{\tiny{H}}} }
\newcommand{\xND}{\x_{\text{\tiny{1:N}}_\text{\tiny{L}}} }

\newcommand{\YN}{\Y_{\mathcal{N}  } }
\newcommand{\YND}{\Y_{\mathcal{L}  } }
\newcommand{\YNA}{\Y_{\mathcal{H}  } }

\newcommand{\NA}{N_\text{\tiny{H}}}
\newcommand{\ND}{N_\text{\tiny{L}}}

\newcommand{\fN}{\f_{\mathcal{N}  } }
\newcommand{\gND}{\g_{\mathcal{L}  } }

\newcommand{\fND}{\f_{\mathcal{L}}}
\newcommand{\fNA}{\f_{\mathcal{H}}}

\newcommand{\sigW}{\sigma^2_\text{\tiny{W}}}

\newcommand{\sigV}{\sigma^2_\text{\tiny{V}}}

\newcommand{\GP}{\mathcal{GP}}
\newcommand{\LGP}{\mathcal{LGP}}

\renewcommand{\Pr}{\mathbb{P}}
\DeclareMathOperator*{\argmin}{arg\,min}
\DeclareMathOperator*{\argmax}{arg\,max}
\renewcommand{\d}{\text{d}}

\newfont{\fsc}{eusm10}                         

\newcommand{\qone}{\bm{q}_1}
\newcommand{\qtwo}{\bm{q}_2}

\newcommand{\Qone}{\bm{Q}_1}
\newcommand{\Qtwo}{\bm{Q}_2}
\newcommand{\Qfour}{\bm{Q}_4}

\newtheorem{theorem}{Theorem}
\newtheorem{lemma}{Lemma}
\newtheorem{remark}{Remark}
\newtheorem{corollary}{Corollary}
\newtheorem{definition}{\noindent \textbf{Definition}}
\begin{document}
\setlength{\textfloatsep}{6pt}
\setlength{\abovecaptionskip}{-1pt}
\setlength{\belowcaptionskip}{-1pt}

\setlength{\parskip}{1pt}
\setlength{\parsep}{1pt}
\setlength{\headsep}{1pt}
\setlength{\topskip}{1pt}
\setlength{\topsep}{1pt}
\setlength{\partopsep}{1pt}
\title{Spatial Field Reconstruction and Sensor Selection in Heterogeneous Sensor Networks with Stochastic Energy Harvesting}
\author{\IEEEauthorblockN{Pengfei Zhang$^{1}$, Ido Nevat$^{2}$, Gareth W.~Peters$^{3}$, Fran\c{c}ois~Septier$^{4}$ and Michael A. Osborne$^{1}$}

$^{1}$ Department of Engineering Science, University of Oxford, UK\\
$^{2}$ TUM CREATE, Singapore\\
$^{3}$ Department of Actuarial Mathematics and Statistics, Heriot-Watt University, Edinburgh, UK \\
$^{4}$ IMT Lille Douai, Univ. Lille, CNRS, UMR 9189 - CRIStAL, F-59000 Lille, France
\\[0pt]
\date{}
}
\maketitle
\begin{abstract}
\noindent
We address the two fundamental problems of \textit{spatial field reconstruction} and \textit{sensor selection} in heterogeneous sensor networks.
We consider the case where two types of sensors are deployed: the first consists of expensive, high quality sensors; and the second, of cheap low quality sensors, which are activated only if the intensity of the spatial field exceeds a pre-defined activation threshold (eg. wind sensors). In addition, these sensors are powered by means of energy harvesting and their time varying energy status impacts on the accuracy of the measurement that may be obtained. We account for this phenomenon by encoding the energy harvesting process into the second moment properties of the additive noise, resulting in a spatial heteroscedastic process.
We then address the following two important problems: (i) how to efficiently perform \textit{spatial field reconstruction} based on measurements obtained simultaneously from both networks; and (ii) how to perform \textit{query based sensor set selection with predictive MSE performance guarantee}.
We first show that the resulting predictive posterior distribution, which
is key in fusing such disparate observations, involves solving intractable
integrals. To overcome this problem, we solve the first problem by developing a low complexity algorithm based on the \textit{spatial best linear unbiased estimator} (S-BLUE).
Next, building on the S-BLUE, we address the second problem, and develop an efficient algorithm for \textit{query based sensor set selection with performance guarantee}. Our algorithm is based on the Cross Entropy method which solves the combinatorial optimization problem in an efficient manner.
We present a comprehensive study of the performance gain that can be obtained by augmenting
the high-quality sensors with low-quality sensors using both synthetic and real insurance storm surge database known as the Extreme Wind Storms Catalogue.
\newline
\textbf{Keywords: } Internet of Things, Sensor Networks, Gaussian Process, Energy harvesting, Cross Entropy method, Sensor selection
\end{abstract}
\section{Introduction}
Wireless Sensor Networks (WSN) have attracted considerable attention due to the large number of applications, such as environmental monitoring \cite{hart2006environmental}, weather forecasts \cite{rajasegarar2014high_j,fonseca2012stability,frenchspatio}, surveillance \cite{sohraby2007wireless}, health care \cite{lorincz2004sensor}, structural safety and building monitoring\cite{chintalapudi2006monitoring} and home automation \cite{frenchspatio, Aky0102}. We consider a WSN which consists of a set of spatially distributed sensors that may have limited resources, such as energy and communication bandwidth. These sensors monitor a spatial physical phenomenon containing some desired attributes (e.g pressure, temperature, concentrations of substance, sound intensity, radiation levels, pollution concentrations, seismic activity etc.) and regularly communicate their observations to a Fusion Centre (FC) \cite{ fazelrandom, matamoros_estimation, vuran2004spatio}. The FC collects these observations and fuses them in order to reconstruct the spatial field \cite{Aky0102}.


In many cases these WSN use a small set of high-quality and expensive sensors (such as weather stations) \cite{VictoriaSensors, peters2015utilize,peters2014modelling}. While these sensors are capable of reliably measuring the environmental physical phenomenon, the low spatial deployment resolution prohibits their use in spatial field reconstruction tasks. To overcome this problem, sparse high-quality sensor deployment can be augmented by the use of complementary cheap low-quality sensors that can be deployed more densely due to their low costs \cite{rajasegarar2014high_j,rajasegarar2014high}.
This type of \textit{heterogeneous sensor networks} approach has gained attention in the last few years due to the vision of the Internet of Things (IoT) where networks may share their data over the internet \cite{gubbi2013internet,vermesan2011internet}.
This coupling enables the concept of \textit{Collaborative Wireless Sensor Network} (CWSN), in which networks with different capabilities are deployed in the same physical region and collaborate in order to optimize various design criteria and processes \cite{perera2014sensor}.
The incentive to develop such heterogeneous sensor network and associated signal processing was further strengthened when the US Environmental Protection Agency (EPA) published its shift in the paradigm of data collection which promotes the notion of augmenting sparse deployments of high-quality sensors with dense deployment of low-quality and inaccurate sensors \cite{EPAreport}.

Two practical scenarios that are of importance are:
\begin{enumerate}
	\item  High-quality sensors may be deployed by government agencies (eg. weather stations). These are sparsely deployed due to their high costs, limited space constraints, high power consumption etc. To improve the coverage of the WSN, low-quality cheap sensors can be deployed to augment the high-quality sensor network \cite{rajasegarar2014high}.
\item High-quality sensors cannot be easily deployed in remote locations, for example in oceans, lakes, mountains and volcanoes. In these cases, energy harvesting based battery operated low-quality cheap sensors can be deployed \cite{werner2006deploying}.
\end{enumerate}
In this paper we consider low-quality sensors which are capable of measuring the intensity of the spatial random field, only if it exceeds a pre-defined threshold. For example, low cost wind sensors are able to measure the wind speed only if it exceeds the activation threshold, see for example \cite{Anemo4403,adafruit}. In addition, these sensors are powered by means of energy harvesting which impacts their reading accuracy. We encode this aspect into the statistical properties of the additive noise term, which results in a spatially correlated heteroscedastic process \cite{wang2012gaussian}. The FC then receives a vector of observations from both the high-quality and low-quality sensors.
Hence, the consequence is that the observations are heterogeneous and generally non-Gaussian distributed as the activation threshold procedure introduces a non-linear transformation of the observations, which makes the data fusion a more complex inference problem.
The main goal of this paper is to develop low complexity algorithms to solve the problems of \textit{spatial field reconstruction} and \textit{query based sensor set selection with performance guarantee} of spatial random fields in WSN under practical scenarios of high and low quality sensors.
%
More specifically, the following two fundamental problems are the focus of this paper:
\begin{enumerate}
	\item  \textit{\textbf{Spatial field reconstruction}:} the task is to accurately estimate and predict the intensity of a spatial random field, not only at the locations of the sensors, but at all locations \cite{peters2015utilize, nevat2015estimation, nevat2013random}, given heterogeneous observations from both sensor networks.
	\item \textit{\textbf{Query based sensor set selection with performance guarantee}}: the task is to perform on-line sensor set selection which meets the QoS criterion imposed by the user, as well as minimises the costs of activating the sensors of these networks \cite{calvo2016sensor,joshi2009sensor,chepuri2015sparsity}.
	\end{enumerate}
\subsection{Related work on spatial field reconstruction in sensor networks:}
It is common in the literature to model the physical phenomenon being monitored by the WSN according to a Gaussian random field (GRF) with a spatial correlation structure  \cite{ akyildiz2004exploiting, gu2012spatial,nevat2012location,7945510,plonski2013energy,6866922,naderi2015wireless,oliveira2016characterization}. More generally, examples of GRFs include wireless channels \cite{agrawal2009correlated}, speech processing \cite{park2008gaussian}, natural phenomena (temperature, rainfall intensity etc.) \cite{kottas2012spatial, fonseca2012stability}, and recently in \cite{nevat2015estimation, nevat2013random}.

The simplest form of Gaussian process model would typically assume that the spatial field observed at the FC is only corrupted by additive Gaussian noise. For example, in \cite{gu2012spatial} a linear regression algorithm for GRF reconstruction in mobile wireless sensor networks was presented, but relied on the assumption of only Additive White Gaussian Noise (AWGN); in \cite{xu2011adaptive} an algorithm was developed to learn the parameters of non-stationary spatio-temporal GRFs again assuming AWGN; and in\cite{krause2008near} an algorithm for choosing sensor locations in GRF assuming AWGN was developed.

In practical WSN deployments, two deviations from these simplified modelling assumptions arise and are important to consider: the use of heterogeneous sensor types i.e. sensors may have different degrees of accuracy throughout the field of spatial monitoring \cite{nevat2015estimation, nevat2013random}; and secondly quite often the sensors may be powered by means of energy harvesting which impacts their reading accuracy and lifetime duration \cite{basagni2013wireless,naderi2015wireless}.

In \cite{calvo2016sensor}, the authors developed an algorithm  for sensor selection and power allocation in energy harvesting wireless sensor networks. They extended the model proposed by Joshi and Boyd \cite{joshi2009sensor} and incorporated the energy harvesting process into the problem formulation. However, they assumed that the energy harvesting process is fully observed, which is too an optimistic assumption in practice.
In \cite{liu2014energy,chepuri2015sparsity}, a sparsity-promoting penalty function to discourage repeated selection of any sensor node was proposed.

All of these works did not attempt to solve the sensor selection problem in heterogeneous sensor networks. Nor did they assume that the energy harvesting process is not fully observed and impacts the reading accuracy of the sensors. Our paper considers these important aspects and provides both holistic and practical solution for those problems.
\subsection{Contributions:}
\begin{enumerate}
	\item We develop a novel statistical model to account for high and low quality sensors (see A3-A4 in Section \ref{SystemModel}).
	\item We model the practical scenario of spatially correlated additive noise due to energy harvesting via a spatially correlated heteroscedastic process (see A5-A6 in Section \ref{SystemModel}).
	\item We develop a point-wise estimation algorithm for spatial field reconstruction which is based on the Spatial Best Linear Unbiased Estimator (S-BLUE) (Corollary \ref{s_blue_algorithm} in Section \ref{S_BLUE}).
		\item We develop an efficient algorithm to perform \textit{query based sensor set selection with performance guarantee} which solves the combinatorial optimization problem using the Cross Entropy method (Section \ref{QuerySensorSetSelection}).
\end{enumerate}
\section{Sensor Network Model and Definitions} \label{SystemModel}
We begin by presenting the statistical model for the spatial physical phenomena, followed by the system model.
\subsection{Spatial Gaussian Random Fields Background}
We model the physical phenomena (both monitored and energy harvesting phenomena) as spatially dependent continuous processes with a spatial correlation structure and are independent from each other. Such models have recently become popular due to their mathematical tractability and accuracy \cite{zhang2015near,plonski2013energy,ling2015gaussian}.
The degree of the spatial correlation in the process increases with the decrease of the separation between two observing locations and can be accurately modelled as a Gaussian random field \footnote{We use Gaussian Process and Gaussian random field interchangeably.}
\cite{nevat2015estimation,nevat2013random, akyildiz2004exploiting,nevat2012location,agrawal2009correlated, fonseca2012stability}.
 A Gaussian process (GP) defines a distribution over a space of functions and it is completely specified by the equivalent of sufficient statistics for such a process, and is formally defined as follows.
\begin{definition}
\label{DefGP}
(Gaussian process \cite{rasmussen2005gaussian},\cite{adler2007random}):
\textsl{Let $\mathcal{X} \subset \mathbb{R}^d$ be some bounded
domain of a d-dimensional real valued vector space. Denote by $f(\x): \mathcal{X} \mapsto  \mathbb{R}$ a stochastic process parametrized by $\x \in \mathcal{X}$. Then, the random function $f(\x)$ is a Gaussian process if all its finite dimensional distributions are Gaussian, where for any $m \in \mathbb{N}$, the random vectors $\left(f\left(\x_1\right),\cdots, f\left(\x_m\right)\right)$ are normally distributed.}
\end{definition}
We can therefore interpret a GP as formally defined by the following class of random functions:
\begin{align*}
\begin{split}
&\mathcal{F} :=
\left\{
f\left(\cdot\right) : \mathcal{X} \mapsto  \mathbb{R}
\;
\text{s.t.}\; f\left(\cdot\right)
\sim\GP\left(\mu \left(\cdot;\THETA\right),\C\left(\cdot,\cdot;\PSI\right)\right), \;\right.\\
&\;\mathrm{with}\; \mu\left(\x;\THETA\right):=\exE\left[f\left(\x\right)\right] : \mathcal{X} \mapsto  \mathbb{R},\\
&\C\left(\x_i,\x_j;\PSI\right):=
\exE\left[\left(f\left(\x_i\right)-\mu\left(\x_i;\THETA\right)\right)
\left(f\left(\x_j\right)-\mu\left(\x_j;\THETA\right)\right)
\right],\\
& : \mathcal{X} \times \mathcal{X} \mapsto  \mathbb{R}^+
\end{split}
\end{align*}	
\normalsize
where at each point the mean of the function is $\mu(\cdot;\THETA),$ parameterised by $\THETA$, and the spatial dependence between any two points is given by the covariance function (Mercer kernel)
$\C \left(\cdot,\cdot;\PSI\right)$, parameterised by $\PSI$, see detailed discussion in \cite{rasmussen2005gaussian}.\\
It will be useful to make the following notational definitions for the cross correlation vector and auto-correlation matrix, respectively:
\footnotesize
\begin{align*}
\begin{split}
k\left(\x_*,\xN\right) &:= \exE\left[f\left(\x_*\right) \; f\left(\xN\right)\right] \in \mathbb{R}^{1 \times \mathcal{N} }\\
\K\left(\xN,\xN\right) &:=
\left[
\begin{array}{ccc}
 \mathcal{C}\left(\x_1,\x_1\right) &  \cdots    & \mathcal{C}\left(\x_1,\x_{\mathcal{N}}\right)\\
 \vdots &  \ddots    & \vdots\\
 \mathcal{C}\left(\x_{\mathcal{N}},\x_1\right) &  \cdots    & \mathcal{C}\left(\x_{\mathcal{N}},\x_{\mathcal{N}}\right)
\end{array}
\right],
\end{split}
\end{align*}	
\normalsize
with $\mathcal{S}^+\left(\mathbb{R}^n\right)$ is the manifold of symmetric positive definite matrices. $k\left(\x_i,\x_j\right)$ is the correlation function and $\mathcal{C}\left(\x_i,\x_j\right)$ is the covariance function.

We define the Log Gaussian Process (LGP) in the following definition:

\begin{definition}
\label{DefGP}
(Log Gaussian process):
A Log Gaussian Process (LGP) defines a stochastic process whose logorithm follows Gaussian Process. In mathematics, given a
\begin{align}
g\left(\x\right) \sim
\LGP\left(\mu_g \left(\x;\THETA_g\right),\C_g\left(\x_1,\x_2;\PSI_g\right)
\right).
\end{align}	
where \[\exE[g\left(\x\right)]=e^{\mu_g \left(\x;\THETA_g\right)+\C_g\left(\x_1,\x_2;\PSI_g\right)^2/2}\]
\end{definition}

Having formally specified the semi-parametric class of Gaussian process models and Log Gaussian process model, we proceed with presenting the system model.

\subsection{Heterogeneous Sensor Network System Model} \label{SystemModel}
We now present the system model for the physical phenomenon observed by two types of networks and the energy harvesting model.
\begin{enumerate}[noitemsep]
\item[A1] Consider a random spatial phenomenon (eg. wind) to be monitored defined over a $2$-dimensional space $\mathcal{X} \in \mathbb{R}^{2}$. The mean response of the physical process is a smooth continuous spatial function $f\left(\cdot\right):\mathcal{X} \mapsto  \mathbb{R}$, and is modelled as a Gaussian Process (GP) according to
\begin{align}
f\left(\x\right) \sim \GP\left(\mu_f \left(\x;\THETA_f\right)
,\C_f\left(\x_1,\x_2;\PSI_f\right)
\right),
\end{align}	
where the mean and covariance functions $\mu_f \left(\x;\THETA_f\right),\C_f\left(\x_1,\x_2;\PSI_f\right)$ are assumed to be known.
\item[A2] Let $N$ be the total number of sensors that are deployed over a $2$-D region $\mathcal{X} \subseteq \mathbb{R}^2$, with
$\x_{n} \in \mathcal{X}, n=\left\{1,\cdots, N\right\}$ being the physical location of the $n$-th sensor, assumed known by the FC. The number of sensors deployed by Network $1$ and Network $2$ are $\NA$ and $\ND$, respectively, so that $N=\NA+\ND$ .

\item[A3] \textbf{Sensor network $1$: High Quality Sensors}\\
The sensors have a $0$-threshold activation and each of the sensors collects a noisy observation of the spatial phenomenon $f\left(\cdot\right)$. At the $n$-th sensor, located at $\x_n$, the observation is given by:
\begin{align}
Y^{H}\left(\x_n\right)= f\left(\x_n\right) + W\left(\x_n\right),  \;n=\left\{1,\cdots,\NA\right\}
\end{align}
where $W\left(\x_n\right)$ is i.i.d Gaussian noise $W\left(\x_n\right) \sim N\left(0,\sigW\right)$.
\item[A4] \textbf{Sensor network $2$: Low Quality Sensors}\\
The sensors have a $T$-threshold activation and each of the sensors collects a noisy observation of the spatial phenomenon $f\left(\cdot\right)$, only if the intensity of the field at that location exceeds the pre-defined threshold $T$, (eg. anemometer sensors for wind monitoring \cite{adafruit,Anemo4403}). At the $n$-th sensor, located at $\x_n$, the observation is given by:
\begin{align}
Y^{L}\left(\x_n\right)=
\left\{
\begin{array}{ll}
f\left(\x_n\right)+V\left(\x_n\right),&f\left(\x_n\right)\geq T\\
V\left(\x_n\right),&f\left(\x_n\right)<T
\end{array}
\right.
\end{align}
The statistical properties of the additive noise $V\left(\x_n\right)$ are detailed in A6.
\item[A5] \textbf{Energy harvesting model:}\\
The energy harvesting process (eg. solar irradiance) is modelled as a spatial phenomenon defined over a $2$-dimensional space $\mathcal{X} \in \mathbb{R}^{2}$. The mean response of the physical process is a smooth continuous spatial function $g\left(\cdot\right):\mathcal{X} \mapsto  \mathbb{R}$, and in order to ensure positivity it is modelled as a $\log$- Gaussian Process (LGP), similar as \cite{7945510,plonski2013energy,naderi2015wireless,oliveira2016characterization}, given by
\begin{align}
g\left(\x\right) \sim
\LGP\left(\mu_g \left(\x;\THETA_g\right),\C_g\left(\x_1,\x_2;\PSI_g\right)
\right).
\end{align}	
We assume that the energy harvesting process $g(\x)$ is independent of the monitored physical phenomenon $f(\x)$.\footnote{The energy harvesting and the physical phenomenon being measured would not be independent processes in some cases. In this paper however, we only consider the case where the energy harvesting model and the physical phenomenon to be independent process, which is of practical importance in many cases. Some examples include energy harvesting via solar power and a physical phenomenon which is precipitation or pollution.}
\item[A6] \textbf{Spatial noise process model:}\\
Since all sensors in Network $2$ use energy harvesting techniques from a spatially correlated energy field $g\left(\x\right)$, this will have an impact on the performance of the electronic circuits (eg. amplifiers, voltage and frequency biases). These impact the thermal noise ($V\left(\x\right)$) and its characteristics, see \cite{zhang2011accurate}. As a result of the variations and fluctuations of the energy field, the additive thermal noise is now also spatially correlated across sensors exposed to common environmental energy harvesting conditions, such that $\mathbb{E}[V(x_1)V(x_2)]\neq 0$. The spatial noise can be modelled as a spatial stochastic volatility model in which the variance of the additive noise is itself a random process which is spatially correlated. This means that if the energy harvested by a particular sensor is high, the noise variance should be small and vice versa. A common approach to modelling this impact is via a link function as follows:
    \begin{align}
    \sigV\left(\x\right)= \psi\left(g\left(\x\right)\right),
    \end{align}
		where $\psi\left(\alpha\right):  \mathbb{R}^{+} \rightarrow \mathbb{R}^{+}$ is a deterministic known mapping. The choice of $\psi(x)$ can be flexible in practice, but it needs to satisfy the constraints that the larger the value of $x$, the smaller $\psi(x)$ is. Some choice of $\psi(x)$ includes $1/x, \exp(-x), 1/x^2$ and etc. In this paper,  without loss of generality and for notational simplicity, in the paper we assume that $\psi\left(\alpha\right)=1/\alpha.$
\end{enumerate}
%
In Table \ref{tab:notations} we present the notations which will be used throughout the paper.
\begin{table*}[htbp]
\caption{Table of Notations }
\centering 
\begin{tabular}{|l| l| }
 \hline
Variable& Meaning\\
\Xhline{2pt}
$\xN$& physical locations in terms of $\left[x,y\right]$ coordinates of the $N$ sensors deployed in the field.\\
\hline
$\YN=\left\{Y_1,\ldots,Y_N\right\} \in \mathbb{R}^{1 \times N}$ & collection of observations from all sensors (both Network $1$ and Network $2$) at the fusion center.\\
\hline
$\YNA \in \mathbb{R}^{1 \times \NA} \subseteq \YN$ & collection of observations from all sensors in Network $1$ at the fusion center.\\
\hline
$\YND \in \mathbb{R}^{1 \times \ND} \subseteq \YN$ & collection of observations from all sensors in Network $2$, at the fusion center.\\
\hline
	$\fN=\left\{f_1,\ldots,f_N\right\} \in \mathbb{R}^{N \times 1}$ & realisation of $\f$ at the sensors located at $\xN$.\\
\hline
	$\fNA \subseteq \fN$ & realisation of $\f$ at the sensors of Network $1$, located at $\xNA \subseteq \xN$.\\
\hline
	$\fND \subseteq \fN$ & realisation of $\f$ at the sensors of Network $2$, located at $\xND \subseteq \xN$.\\
\hline
  $\gND=\left\{g_1,\ldots,g_{N_L}\right\} \in \mathbb{R}^{N_L \times 1}$ & realisation of the energy field $\G$ at the sensors of in Network $2$ located at $\xND$.\\
\hline
\end{tabular}
\label{tab:notations}
\end{table*}


\section{Field Reconstruction via Spatial Best Linear Unbiased Estimator (S-BLUE)}
To perform inference in our Bayesian framework, one would typically be interested in computing the predictive posterior density at any location in space, $\x_* \in  \mathcal{X} $ , denoted $p\left(f_*|\YN\right)$. Based on this quantity various point estimators, like the Minimum Mean Squared Error (MMSE) and the Maximum A-Posteriori (MAP) estimators can be derived:
\begin{align*}
\begin{split}
\widehat{f_*}^{\text{MMSE}} &= \int \limits_{-\infty}^{\infty}
p\left(f_*|\YN,\xN,\x_*\right) f_* \d f_*,\\
\widehat{f_*}^{\text{MAP}} &= \argmax_{f_*} p\left(f_*|\YN,\xN,\x_*\right).
\end{split}
\end{align*}
These estimators provide a pointwise estimator of the intensity of the spatial field, $\widehat{f_*}$ at location $\x_*$. This enables us to reconstruct the whole spatial field by evaluating $\widehat{f_*}$ on a fine grid of points.
The predictive posterior density is given by:
\footnotesize
\begin{align}
\begin{split}
&p\left(f_*|\YN,\xN,\x_*\right)= \int\limits_{\mathcal{R}^N} p\left(f_*|\fN,\xN,\x_*\right)p\left(\fN|\YN,\xN,\x_*\right) \d \fN\\
&=					 \int\limits_{\mathbb{R}^N}  p\left(f_*|\fN,\xN,\x_*\right)p\left(\fNA,\fND|\YNA,\YND,\xN\right) \d \fN\\
&=					 \int\limits_{\mathbb{R}^N} p\left(f_*|\fN,\xN,\x_*\right) \int\limits_{\mathbb{R}^{\ND}} p\left(\fNA|\YNA,\YND,\fND,\xN,\gND\right)\\
						&\times  p\left(\fND|\YNA,\YND,\xN,\gND\right)						
						p\left(\gND|\YNA,\YND,\xN,\gND\right)\d \gND \d \fN\\
&=					 \int\limits_{\mathbb{R}^N}  p\left(f_*|\fN,\xN,\x_*\right)p\left(\fNA|\YNA,\YND,\fND,\xN,\gND\right)\\
						&\times \int\limits_{\mathbb{R}^{\ND}}
						\frac{p\left(\YN|\fND,\xN\right)p\left(\fND|\xND\right)}{\int\limits_{\mathcal{R}^{\ND}}  p\left(\YN|\fND,\xN,\gND\right)p\left(\fND|\xND\right)\d\fND}
						p\left(\gND|\YNA,\YND,\xN\right)\\&\d \gND \d \fN.
\end{split}
\end{align}
\normalsize
where $\mathbb{R}^{\ND}$ defines the domain for $\gND$, which is the $\ND$ dimensional LGP specified in A5 in Section \ref{SystemModel}. Unfortunately, the predictive posterior density cannot be calculated analytically in closed form, prohibiting a direct calculation of any Bayesian estimator.
One approach to approximating $p\left(f_*|\YN\right)$ is the Laplace approximation \cite{daniels1954saddlepoint}, which was used in our previous works \cite{nevat2015estimation, nevat2013random}. A different approach is based on Markov Chain Monte Carlo (MCMC) methods which generates samples from the target distribution $p\left(f_*|\YN,\xN,\x_*\right)$ \cite{wang2012gaussian}.
These methods are not suitable for our problem as we require low complexity algorithm which is suitable for reconstructing the whole spatial field as well as for selecting the optimal subset of sensors in real-time. In addition, expectation propagation and variational inference are unsuitable due to computational constraints. To achieve these goals, we develop a low-complexity linear estimator for $f_*$, presented next.
\subsection{Spatial Best Linear Unbiased Estimator (S-BLUE) Field Reconstruction Algorithm} \label{S_BLUE}
We develop the spatial field reconstruction via Best Linear Unbiased Estimator (S-BLUE), which enjoys a low computational complexity \cite{kay:1998}. The S-BLUE does not require calculating the predictive posterior density, but only the first two cross moments of the model.
The S-BLUE is the optimal (in terms of minimizing Mean Squared Error (MSE)) of all linear estimators and is given by the solution to the following optimization problem:
\begin{equation}
\label{S_BLUE_objective}
\widehat{f}_* :=\widehat{a} + \widehat{\B} \YN = \arg \min_{a, \B} \exE\left[\left(f_*- \left(a + \B \YN\right)\right)^2\right],
\end{equation}
where $\widehat{a} \in \mathbb{R}$ and $\widehat{\B} \in \mathbb{R}^{1 \times N}$.

The optimal linear estimator that solves (\ref{S_BLUE_objective}) is given by
\begin{align}
\begin{split}
\label{s_blue_estimate}
\hat{f_*}&=\exE_{f_*\; \YN}\left[f_* \; \YN\right]\exE_{\YN}\left[ \YN\; \YN\right]^{-1}\left(\YN-\exE\left[\YN\right]\right),
\end{split}
\end{align}
and the Mean Squared Error (MSE) is given by
\begin{align}
\begin{split}
\label{s_blue_estimate_MSE}
\sigma^2_{*}&=k\left(\x_*,\x_*\right)-
\exE_{f_*\; \YN}\left[f_*\; \YN\right]\exE_{\YN}\left[\YN\; \YN \right]^{-1}
\\&\times\exE_{\YN\;f_*}\left[\YN \; f_*\right].
\end{split}
\end{align}

\begin{remark}
Non-zero mean random spatial phenomenon: in order to handle the practical case where $\mu_f \left(\x;\THETA_f\right) \neq 0$, we first subtract this known value from our observations $Y^{H}\left(\x_n\right), $ and $Y^{L}\left(\x_n\right)$. We then apply our algorithm, and finally shift back our results by $\mu_f \left(\x;\THETA_f\right).$
\end{remark}

To evaluate (\ref{s_blue_estimate}-\ref{s_blue_estimate_MSE}) we need to calculate the cross-correlation $\exE_{f_*, \YN}\left[f_*\; \YN\right]$, auto-correlation $\exE_{\YN}\left[\YN\; \YN^T \right]$ and $\exE\left[\YN\right]$.

\subsection{Cross-correlation between a test point and sensors observations $\exE_{f_*, \YN}\left[f_*\; \YN\right]$}
To calculate the cross-correlation vector, we decompose the observation vector into its high and low quality observations, given by:
\begin{align}
\begin{split}
\YN =\left[\YNA ,\; \YND\right],
\end{split}
\end{align}
and express the cross-correlation vector $\exE_{f_*, \YN}\left[f_*\; \YN\right]$ in (\ref{s_blue_estimate}-\ref{s_blue_estimate_MSE}) as:
\begin{align}
\begin{split}
\label{cross_correlation}
\exE_{f_*, \YN}\left[f_*\; \YN\right]
=\big[
\underbrace{\exE_{f_*, \YNA}\left[f_*,  \YNA \right]}_{\mbox{\large $\qone$}},
\underbrace{\exE_{f_*, \YND}\left[f_*, \YND\right]}_{\mbox{\large $\qtwo$}}
\big]
\end{split}
\end{align}
We can now derive separately the cross-correlation between the test point $\x_*$ and a sensor observation from Network $1$ and Network $2$.
We begin with the derivation of the cross correlation between a test point $f_*$ and the observations from sensors in Network $1$, $\YNA$, presented in Lemma \ref{lemma_cross_correlation_1} followed by the cross correlation between a test point $f_*$ and observations from sensors in Network $2$, $\YND$, presented in Lemma \ref{lemma_cross_correlation_2}.
\begin{lemma} (Calculating $\exE_{f_*, \YNA}\left[f_*\;  \YNA \right]$):$\;$\\
\label{lemma_cross_correlation_1}
The cross correlation between a test point $f_*$ and the $k$-th $\left(k = \left\{1,\ldots, N_H\right\}\right)$ sensor observation in Network $1$ is given by:
\begin{align*}
\left[\mbox{\large $\qone$}\right]_k:= \exE_{f_*, Y_H\left(\x_k\right)}\left[f_*\; Y_H\left(\x_k\right)\right]= k_f\left(\x_*,\x_k\right).
\end{align*}
\end{lemma}
\begin{proof}
See Appendix \ref{lemma_cross_correlation_1_proof}
\end{proof}
\begin{lemma} (Calculating $\exE_{f_*, \YND}\left[f_*\;  \YND \right]$):$\;$\\
\label{lemma_cross_correlation_2}
The cross correlation between a test point $f_*$ and the $k$-th $\left(k = \left\{1,\ldots, N_L\right\}\right)$ sensor observation in Network $2$ is given by:

\begin{align*}
&\left[\mbox{\large $\qtwo$}\right]_k:=
\exE_{f_*, Y_L\left(\x_k\right)} \left[f_*\; Y_L\left(\x_k\right)\right]\\
&=\mathcal{C}_f\left(\x_*,\x_k\right)\\&\times\Biggl(1-\Phi\left(\frac{T}{\sqrt{\mathcal{C}_f\left(\x_k,\x_k\right)}}\right)+\frac{T}{\sqrt{\mathcal{C}_f\left(\x_k,\x_k\right)}}\phi\left(\frac{T}{\sqrt{\mathcal{C}_f\left(\x_k,\x_k\right)}}\right)\Biggr).
\end{align*}
\end{lemma}

\begin{proof}
See Appendix \ref{lemma_cross_correlation_2_proof}
\end{proof}
Using Lemmas \ref{lemma_cross_correlation_1}-\ref{lemma_cross_correlation_2}, the cross-correlation vector $\qone$, and $\qtwo$ in (\ref{cross_correlation}) is derived.
\subsection{Correlation of sensors observations $\exE_{\YN}\left[\YN \; \YN^T \right]$}
We now derive $\exE_{\YN}\left[\YN\; \YN^T \right]$, which is the correlation matrix of all sensors observations from both Network $1$ and Network $2$.
This involves the correlation within Network $1$ and Network $2$ and across the networks, given by
\footnotesize
\begin{align}
\begin{split}
&\exE_{\YN}\left[\YN\; \YN^T \right]=
\exE_{\YNA,\YND}\left[\left[\YNA \; \YND^T\right], \left[\YNA, \; \YND\right] \right]\\
& =
\left[
    \begin{array}{c;{2pt/2pt}r}
    \exE_{\YNA}\left[\YNA \; \YNA^T\right] & \exE_{\YNA,\YND}\left[\YNA\; \YND^T\right]\\
		\hdashline[2pt/2pt]
    \exE_{\YNA,\YND}\left[\YND \; \YNA^T\right] & \exE_{\YND}\left[\YND \; \YND^T\right]
    \end{array}
\right]\\
& :=
\left[
    \begin{array}{c;{2pt/2pt}r}
    \mbox{\LARGE $\Qone$} & \mbox{\LARGE $\Qtwo$}\\
		\hdashline[2pt/2pt]
    \mbox{\LARGE $\Qtwo^T$} & \mbox{\LARGE $\Qfour$}
    \end{array}
\right].
\end{split}
\end{align}
\normalsize
\begin{lemma} (Calculating $\exE_{\YNA}\left[\YNA \;\YNA^T\right]$):$\;$\\
\label{lemma_cross_correlation_3}
The correlation between a sensor observation in Network $1$, $Y_H\left(\x_k\right) \in \YNA $ and a sensor observation in Network $1$, $Y_H\left(\x_j\right)\in \YNA $ is given by
\begin{align*}
\left[\Qone\right]_{k,j}&:=
\exE_{Y_H\left(\x_k\right), Y_H\left(\x_j\right)} \left[Y_H\left(\x_k\right) Y_H\left(\x_j\right)\right]\\ &=
k_f\left(\x_k,\x_j\right)+\mathds{1}\left(k=j\right)\sigW.
\end{align*}
\end{lemma}
\begin{proof}
See Appendix \ref{lemma_cross_correlation_3_proof}
\end{proof}

\begin{remark}
Sensors of high quality have spatially uncorrelated thermal noise as no energy harvesting is required, and the Dirac measure on the diagonal when $k=j$.
\end{remark}
\begin{lemma} (Calculating $\exE_{\YNA,\YND}\left[\YNA, \YND\right]$):$\;$\\
\label{lemma_cross_correlation_4}
The correlation between a sensor observation in Network $1$, $Y_H\left(\x_k\right) \in \YNA $ and a sensor observation in Network $2$, denoted  $Y_L\left(\x_j\right)\in \YND $ is given by
\footnotesize
\begin{align*}
&\left[\Qtwo\right]_{k,j}:=
\exE_{Y_H\left(\x_k\right), Y_L\left(\x_j\right)}\left[Y_H\left(\x_k\right) \; Y_L\left(\x_j\right)\right]\\
 &=\mathcal{C}_f\left(\x_k,\x_j\right)\Biggl(1-\Phi\left(\frac{T}{\sqrt{\mathcal{C}_f\left(\x_j,\x_j\right)}}\right)\\&+\left(\frac{T}{\sqrt{\mathcal{C}_f\left(\x_j,\x_j\right)}}\right)
\phi\left(\frac{T}{\sqrt{\mathcal{C}_f\left(\x_j,\x_j\right)}}\right)\Biggr).
\end{align*}
\normalsize
\end{lemma}
\begin{proof}
See Appendix \ref{lemma_cross_correlation_4_proof}
\end{proof}
Next we derive $\Qfour$, where we separate this calculation into two cases: the diagonal elements of $\Qfour$ are calculated in Lemma \ref{cross_correlation_6}, and the non-diagonal elements in Lemma \ref{theorem_cross_correlation_7}.
\begin{lemma} (Calculating diagonal elements of $\exE_{\YND}\left[\YND \YND^T\right]$):$\;$\\
\label{cross_correlation_6}
The auto-correlation of sensor observations in Network $2$, $Y\left(\x_k\right) \in \YND $ is given by
\footnotesize
\begin{align*}
\begin{split}
&\left[\Qfour\right]_{k,k}:=
\exE_{Y_L\left(\x_k\right)} \left[Y_L\left(\x_k\right) \; Y_L\left(\x_k\right)\right]\\
&=\mathcal{C}_f\left(\x_k,\x_k\right)\Biggl(1-\Phi\left(\frac{T}{\sqrt{\mathcal{C}_f\left(\x_k,\x_k\right)}}\right)+\left(\frac{T}{\sqrt{\mathcal{C}_f\left(\x_k,\x_k\right)}}\right)\\&\times\phi\left(\frac{T}{\sqrt{\mathcal{C}_f\left(\x_k,\x_k\right)}}\right)\Biggr)+\exp\left(\mu_g\left(\x_k\right) +
 \frac{\mathcal{C}_g\left(\x_k,x_k\right)}{2}\right).
\end{split}
\end{align*}
\end{lemma}
\normalsize
\begin{proof}
See Appendix \ref{theorem_cross_correlation_6_proof}.
\end{proof}
We now consider the calculation of the correlation between a single sensor observation in Network $2$, $Y_L\left(\x_k\right) \in \YND $ and a sensor observation in Network $2$, $Y_L\left(\x_j\right) \in \YND$ at different locations, (ie. non-diagonal elements $\x_k \neq \x_j$). This is given by
\begin{align*}
\begin{split}
\left[\Qfour\right]_{k,j}&:=
\exE_{Y_L\left(\x_k\right),Y_L\left(\x_j\right)} \left[Y_L\left(\x_k\right) \; Y_L\left(\x_j\right)\right].
\end{split}
\end{align*}
To obtain this result, we first present the following Theorem which states useful results regarding the correlation of bi-variate truncated Normal random variables.
\begin{theorem}[Correlation of Bivariate Truncated Normal Random Variables \cite{begier1971correlation}]$\;$\\
\label{theorem_hermite}
Given the standardized bivariate Gaussian distribution $\Z = \left[Z_1,Z_2\right]$, where
$\exE\left[Z_1\;Z_2\right] = \rho$, and observations available only inside the region $\left[a \leq z_1 < \infty, b \leq z_2 < \infty \right]$, then
the cross correlation $\exE[Z_1 Z_2]$ is given by:
\footnotesize
\begin{equation*}
\begin{split}
&\exE\left[Z_1\; Z_2\right] =
\int\limits_a^{\infty}\int\limits_b^{\infty} p\left(z_1,z_2\right) z_1 z_2 \d z_1\; \d z_2
=\rho \Bigl(a\phi\left(a\right)\left(1-\Phi\left(A\right)\right)\\&+ b\phi\left(b\right)\left(1-\Phi\left(B\right)\right) + \Omega_{a,b}\Bigr) + (1-\rho^2)f_{\Z}([a,b]; \rho),
\end{split}
\end{equation*}
\normalsize
where $A = \frac{\left(b-\rho a\right)}{\sqrt{1-\rho^2}}$, $B = \frac{\left(a-\rho b\right)}{\sqrt{1-\rho^2}}$, and $\Omega_{a,b} :=\Pr\left(Z_1\geq a \cap Z_2\geq b\right)$ is the joint complementary cumulative distribution function (CCDF), given by:
\footnotesize
\begin{equation*}
\begin{split}
&\Omega_{a,b} :=\Pr\left(Z_1\geq a \cap Z_2\geq b\right)\\
&= \int\limits_a^{\infty}\int\limits_b^{\infty}f_{\Z}(\z; \rho) \d z_1 \d z_2\\
&= \left(1-\Phi\left(a\right)\right) \left(1-\Phi\left(b\right)\right) +
\phi\left(a\right)\phi\left(b\right)
\sum \limits_{n=1}^{\infty}\frac{\rho^n}{n!} H_{n-1}\left(a\right) H_{n-1}\left(b\right),
\end{split}
\end{equation*}
\normalsize
where $\Phi\left(\cdot\right)$ the distribution function of a standard Gaussian,
and  $H_n\left(z\right)$ are the Hermite-Chebyshev polynomials orthogonal to the standardized normal distribution such that
$$\int\limits_{-\infty}^\infty {\mathit{He}}_m(z) \mathit{He}_n(z)\, e^{-\frac{z^2}{2}} \, \mathrm{d}z =
 \sqrt{2 \pi} n! \delta_{nm},$$
and
\begin{equation*}
\begin{split}
\mathit{He}_n(z)&=(-1)^n e^{\frac{z^2}{2}}\frac{d^n}{dz^n}e^{-\frac{z^2}{2}}\\&=\left(z-\frac{d}{dz} \right )^n \cdot 1,
\end{split}
\end{equation*}
or explicitly as
$$ He_n(z) = n! \sum_{m=0}^{\lfloor \tfrac{n}{2} \rfloor} \frac{(-1)^m}{m!(n - 2m)!} \frac{z^{n - 2m}}{2^m}.$$
\end{theorem}
%
Using Theorem \ref{theorem_hermite}, we now derive the non-diagonal elements of $\Qfour$, presented in the following Lemma:
\begin{lemma} (Calculating non-diagonal elements of $\exE_{\YND}\left[\YND, \YND\right]$):$\;$\\
\label{theorem_cross_correlation_7}
The correlation between a single sensor observation in Network $2$, $Y_L\left(\x_k\right) \in \YND $ and a sensor observation in Network $2$, $Y_L\left(\x_j\right) \in \YND$ at different locations, (ie. $\x_k \neq \x_j$) is given by:
\footnotesize
\begin{align*}
\begin{split}
&\left[\Qfour\right]_{k,j}=
\exE_{Y_L\left(\x_k\right),Y_L\left(\x_j\right)} \left[Y_L\left(\x_k\right) \; Y_L\left(\x_j\right)\right]\\
&=
\sqrt{\mathcal{C}_f\left(\x_k,\x_k\right)\mathcal{C}_f\left(\x_j,\x_j\right)}
\mathcal{C}_f\left(\x_k,\x_j\right)\\
&\times
 \left(T_k\phi\left(T_k\right)\left(1-\Phi\left(A\right)\right)+ T_j\phi\left(T_j\right)\left(1-\Phi\left(B\right)\right) + \Omega\right) \\&+ (1-\mathcal{C}^2_f\left(\x_k,\x_j\right))f_{\Z}([T_k,T_j]; \mathcal{C}^2_f\left(\x_k,\x_j\right)).
\end{split}
\end{align*}
where $A = \frac{\left(T_j-\mathcal{C}_f\left(\x_k,\x_j\right) T_k\right)}{\sqrt{1-\mathcal{C}^2_f\left(\x_k,\x_j\right)}}$, $B = \frac{\left(T_k-\mathcal{C}_f\left(\x_k,\x_j\right) T_j\right)}{\sqrt{1-\mathcal{C}^2_f\left(\x_k,\x_j\right)}}$.
\end{lemma}
\normalsize
\begin{proof}
See Appendix \ref{thereom_cross_correlation_7_proof}
\end{proof}

\subsection{Expected value of the observations $\exE_{\YN}\left[\YN\right]$}
Finally, we need to derive the expected value of the observations for the high and low quality sensors.
The expected value of the $k$-th observations for a high quality sensor is given by:

\begin{align}
\label{mean_y}
\exE\left[Y\left(\x_k\right)\right]=\exE \left[Y\left(\x_k\right)+W_{\x_k}\right]=
0.
\end{align}
\begin{lemma}
The expected value of the $k$-th observations for a low quality sensor is presented in the following Lemma.
\label{lemma_mean_y}
\begin{equation*}
\exE[Y\left(\x_k\right)]=\exE_{\sigV}\left[\exE\left[Y(\x_k)|\sigV\right]\right]
=\sqrt{\mathcal{C}_f\left(\x_k,\x_k\right)}\phi\left(T_k\right).
\end{equation*}

\end{lemma}

\begin{proof}
See Appendix \ref{lemma_mean_y_proof}
\end{proof}
We now use the results we derived to express the S-BLUE for the spatial field reconstruction in (\ref{s_blue_estimate}) and the associated MSE in (\ref{s_blue_estimate_MSE}):
\begin{corollary}
\label{s_blue_algorithm}
The S-BLUE spatial field reconstruction at location $\x_*$ is given by:
\tiny
\begin{align*}
\begin{split}
\hat{f_*}&=\left[\mbox{\large $\qone$} \; \; \mbox{\large $\qtwo$}\right]
\left[
    \begin{array}{c;{2pt/2pt}r}
    \mbox{\LARGE $\Qone$} & \mbox{\LARGE $\Qtwo$}\\
		\hdashline[2pt/2pt]
    \mbox{\LARGE $\Qtwo^T$} & \mbox{\LARGE $\Qfour$}
    \end{array}
\right]^{-1}
\left(\left[
    \begin{array}{c}
    \mbox{\large $\YNA$}\\
    \mbox{\large $\YND$}
    \end{array}
\right]
-\left[
    \begin{array}{c}
    \mbox{\large $\exE[\YNA]$}\\
    \mbox{\large $\exE[\YND]$}
    \end{array}
\right]\right).
\end{split}
\end{align*}
\normalsize
The predictive variance is given by
\begin{align}
\begin{split}
\label{s_blue_algorithm_MSE}
\sigma^2_*&=
\C(\x_*,\x_*)-\left[\mbox{\large $\qone$} \; \; \mbox{\large $\qtwo$}\right]
\left[
    \begin{array}{c;{2pt/2pt}r}
    \mbox{\LARGE $\Qone$} & \mbox{\LARGE $\Qtwo$}\\
		\hdashline[2pt/2pt]
    \mbox{\LARGE $\Qtwo^T$} & \mbox{\LARGE $\Qfour$}
    \end{array}
\right]^{-1}
\left[\mbox{\large $\qone$} \; \; \mbox{\large $\qtwo$}\right]^T,
\end{split}
\end{align}
where $\qone$, $\qtwo$, $\Qone$, $\Qtwo$, $\Qfour$ are given in Lemma \ref{lemma_cross_correlation_1}, Lemma \ref{lemma_cross_correlation_2}, Lemma \ref{lemma_cross_correlation_3}, Lemma \ref{lemma_cross_correlation_4}, Lemma \ref{cross_correlation_6} and , Lemma \ref{theorem_cross_correlation_7}, and $\exE[\YNA]$ and $\exE[\YND]$ are given in Eq. (\ref{mean_y}) and Lemma \ref{lemma_mean_y}.
\end{corollary}
There are several quantities in this expression, $\qone$, $\qtwo$, $\Qone$, $\Qtwo$, $\Qfour$. Each of these quantity will play important role in determining the $\hat{f_*}$. $\qone$ specifies the cross correlation between a test point $f_*$ and the $k$-th $\left(k = \left\{1,\ldots, N_H\right\}\right)$ sensor observation in Network $1$, the higher this value, the larger effect high quality sensor observation have on $\hat{f_*}$. Similarly as $\qtwo$, it specifies the cross correlation between a test point and sensor observation in Network $2$.
Specifically, in summary, the close the test point to the sensor location in either Network 1 or Network 2, the more effect will the sensor have on the estimated quantity at the test location.
\section{Query Based Sensor Set Selection with Performance Guarantee }\label{QuerySensorSetSelection}
In this Section we develop an algorithm to perform on-line sensor set selection in order to meet the requirements of a query made by users of the system. In this scenario users can prompt the system and request the system to provide an estimated value of the spatial random field at a location of interest $\x_*$. The user also provides the required allowed error, quantified by the Mean Squared Error (MSE) of the S-BLUE in Eq. (\ref{s_blue_algorithm_MSE}).
This means that the input to the system is a pair of values indicating the location of interest, denoted by $\x_*$ and the maximal allowed uncertainty, denoted by $\sigma^2_q$. Our algorithm will then choose a subset of sensors from both networks to activate in such a way that meets the QoS criterion (maximal allowed uncertainty) as well as minimises the costs of activating the sensors of these networks.
It is important to note that the MSE at any location can be evaluated without taking any measurements, see (\ref{s_blue_algorithm_MSE}). This means that our algorithm for choosing which sensors to activate does not require the sensors to be activated beforehand.
We now formulate the generic sensor set selection problem where the sensors from both Network $1$ and Network $2$ are candidates for activation.
We first define the user's query:
\begin{definition} \label{def:UQ}
(User's Query):\\
\textsl{
A \textit{User's Query} consists of a 2-tuple $Q:= \left(\x_*,\varepsilon\right)$, where
\begin{enumerate}
	\item $\x_* \in \mathcal{X} \subseteq \mathbb{R}^{2}$ represents the location at which the user is interested in estimating the quantity of interest, denoted $\widehat{f}\left(\x_*\right)$.
	\item $\varepsilon \in \mathbb{R}^+$ represents the maximum statistical error which the user is willing to allow for the estimation of the quantity of interest at $\x_*$, quantified by the MSE: $\exE\left[\left(\widehat{f}\left(\x_*\right)-f\left(\x_*\right)\right)^2\right],$ given in (\ref{s_blue_algorithm_MSE}).
\end{enumerate}
Based on the query $Q$, the network outputs a report $R:= \left(\widehat{f}\left(\x_*\right),\sigma_*\right)$, with the constraint $\sigma_* \leq \varepsilon$. If this condition cannot be met, the network reports a Null value and does not activate any sensor configuration.
}
\end{definition}
We defined the activation sets of the sensors in both networks by
$\mathcal{S}_1 \in \left\{0,1\right\}^{\left|\NA \right|}, \mathcal{S}_2 \in \left\{0,1\right\}^{\left|\ND \right|}$. Then the sensor selection problem can be formulated as follows:
\begin{align}
\label{SSO}
\begin{split}
\mathcal{S} &= \argmin_{
\left(
\stackrel{\mathcal{S}_1 \in \left\{0,1\right\}^{\left|\NA \right|}}{\mathcal{S}_2 \in \left\{0,1\right\}^{\left|\ND \right|}}
\right)}
w_h \left|\mathcal{S}_1\right|+w_l \left|\mathcal{S}_2\right|,
\\
&\text{s.t.} \;\;\sigma^2_*< \sigma^2_q,
\end{split}
\end{align}
where $\sigma^2_q$ is the maximal allowed uncertainty at the query location $\x_*$, and $w_h$ and $w_l$ are the known costs of activating a sensor from Network $1$ and Network $2$, respectively.
This optimization problem is not convex,  due to the non-convex Boolean constraints
$\mathcal{S}_1 \in \left\{0,1\right\}^{\left|\NA \right|}, \mathcal{S}_2 \in \left\{0,1\right\}^{\left|\ND \right|}$.
Solving this optimization problem involves exhaustive evaluation of all possible combinations of sensor selections which is impractical for real-time applications. Previous methods to solve such optimization problems in sensor selection involved a relaxation of the non-convex constraint, see for example \cite{calvo2016sensor,joshi2009sensor,chepuri2015sparsity}.
These approaches provide sub-optimal solutions and their theoretical properties are not well understood. We take a different approach for solving the non-convex problem which does not involve relaxation, but instead utilize a stochastic optimization technique, known as the Cross Entropy Method (CEM). The CEM was first proposed by Rubinstein in
$1999$~\cite{rubinstein1999cross} for rare event simulation, but was adapted for solving estimation and optimization problems see \cite{rubinstein1999cross,DeBoer2005}.
We now present a short overview of the CEM, for more details see \cite{rubinstein1999cross,DeBoer2005}. We then develop the algorithm to solve the optimization problem in (\ref{SSO}).

\subsection{Cross Entropy Method for Optimization problems}
Suppose we wish to maximize a function $U\left(\x\right)$ over some set $\mathscr{X}$. Let us denote the maximum by $\gamma^*$; thus,
\begin{align}
\label{CEMO}
\gamma^*=\max_{\x\in\mathscr{X}}U(\x).
\end{align}
The CEM solves this optimization problem by casting the original problem (\ref{CEMO}) into an estimation problem of rare-event probabilities. By doing so, the CEM aims to locate an optimal parametric sampling distribution, that is, a probability distribution on $\mathscr{X}$, rather than locating the optimal solution directly. To this end, we define a collection of indicator functions $\{\mathds{1}_{\{S(\x)\geq \gamma\}}\}$ on $\mathscr{X}$ for various levels $\gamma\in \mathbb{R}$. Next, let $\{f(\cdot ;\V), \V\in \mathscr{V}\}$ be a family of probability densities on $\mathscr{X}$ parametrized by a real-valued parameter vector $\v$. For a fixed $u\in \mathscr{V}$ we associate with (\ref{CEMO}) the problem of estimating the rare-event probability
\begin{align}\label{CEME}
l(\gamma)=\mathbb{P}_u(U(\x)\geq \gamma)=\exE_u\left[\mathds{1}_{\left\{U(\x)\geq \gamma\right\}}\right],
\end{align}
where $\mathbb{P}_u$ is the probability measure under which the random state $\x$ has a discrete pdf $f\left(\cdot ;\V\right)$ and $\exE_u$ denotes the corresponding expectation operator. For a detailed exposition of the CEM, see \cite{rubinstein1999cross,DeBoer2005}.
The CE method involves the following iterative procedure shown in Algorithm \ref{alg:CE}:
\begin{algorithm}[H]
\caption{CE Method}
\label{alg:CE}
\begin{algorithmic}[1]
    \STATE Initialization: Choose an initial parameter vector $\V$
    \WHILE{stopping criterion}
        \STATE  Generate $K$ samples:$\mathbf{\Gamma_i},$
				where $\mathbf{\Gamma_i} \sim f\left(\cdot; \V_t\right)$; $1 \leq i \leq K$.
        \STATE  Evaluate $U\left(\mathbf{\Gamma_i}\right)$ for all the $K$ samples.
        \STATE  Calculate $\beta_t=(1-\rho)$ quantile of $U_1, \ldots, U_K$
        \STATE  Solve the stochastic program to update the parameter vector $\V$:
        \begin{align*}
            \V_t = \argmax_{\V} \frac{1}{K}\sum \limits_{i=1}^K
						\mathds{1}\left(U\left(\mathbf{\Gamma_i}\right) \geq \beta_t \right) \ln \left(f\left(\mathbf{\Gamma_i}; \V_t\right)\right)
        \end{align*}
    \ENDWHILE
\end{algorithmic}
\end{algorithm}
The most challenging aspect in applying the CEM is the selection of an appropriate class of parametric sampling densities $ f\left(\cdot; \V\right), \V \in \mathcal{V}$. There is not a unique parametric family and the selection is guided by competing objectives. The class $f\left(\cdot; \V\right), \V \in \mathcal{V}$ has to be flexible enough to include a reasonable parametric approximation to the optimal importance sampling density. The density $f\left(\cdot; \V\right), \V \in \mathcal{V}$ has to be simple enough to allow fast random variable generation and closed-form solutions to the optimization problem. In addition, to be able to analytically solve the stochastic program, then $f\left(\cdot; \V\right)$ should be a member of the Natural Exponential Families (NEF) of distributions. Under NEFs, the optimization problem can be solved analytically in closed form making the CE very easy to implement \cite{DeBoer2005}.
\subsection{Cross Entropy Method for Sensor Set Selection}
To apply the CEM to solve our optimization problem in (\ref{SSO}), we need to choose a parametric distribution. Since the activation of the sensors is a binary variable (eg. $0 \rightarrow \text{don't activate},1 \rightarrow \text{activate}$), we choose an independent Bernoulli variable as our parametric distribution, with a single parameter $p$ (ie. $\V=p$). The Bernoulli distribution is a member of the NEF of distributions, hence, an analytical solution of the stochastic program is available in closed form as follows:
\begin{align*}
p_{t,j} =\frac{
						\sum\limits_{i=1}^{K}
						\mathds{1}\left(\mathbf{\Gamma^H_{i,j}}=1\right)
						\mathds{1}\left(U\left(k\right) \geq \beta_t \right)
									}
						{\sum\limits_{i=1}^{K}
						\mathds{1}\left(U\left(k\right) \geq \beta_t \right)}.
\end{align*}
Since the optimization problem in Eq. (\ref{SSO}) is a constrained optimization problem, we introduce an Accept$\setminus$Reject step which rejects samples which do not meet the QoS criterion $\sigma^2_*< \sigma^2_q$, as follows
\begin{align*}
U\left(k\right)=
\begin{cases}
-\left(w_h \left|\mathcal{S}^H\right|+w_l \left|\mathcal{S}^L\right|\right),&\sigma^2_*\left(k\right) <\epsilon\\
-\infty,&\text{Otherwise}
\end{cases}
\end{align*}
 The resulting algorithm is presented in Algorithm \ref{alg:CE_generic}.
\begin{algorithm}
\caption{Sensor Selection in Heterogeneous Sensor Networks via Cross Entropy method}
\label{alg:CE_generic}
\begin{algorithmic}
\REQUIRE User's query $Q:= \left(\x_*,\varepsilon\right)$, $\alpha$, $w_h$, $w_l$ and $\Psi$
\STATE 0. Initialization at iteration $t=0$: set $\mathbf{p}^{H}_0=\{p^H_{0,1},p^H_{0,2}, \cdots ,p^H_{0,\NA}\}$ such that $p^H_{0,j}=0.5$, and set $\mathbf{p}^{L}_0=\{p^L_{0,1},p^L_{0,2}, \cdots ,p^L_{0,\ND}\}$ such that $p^L_{0,j}=0.5$.
\WHILE{stopping criterion}
\STATE 1.  Generate $K$ independent samples of binary sets \\
$\mathbf{\Gamma}_i^H=\{\gamma^H_{i,1},\gamma^H_{i,2} \cdots ,\gamma^H_{i,\NA}\}$, where $\gamma^H_{i,j} \sim Ber\left(p^H_{t,j}\right)$; $1 \leq i \leq K$ and\\
$\mathbf{\Gamma}_i^L=\{\gamma^L_{i,1},\gamma^L_{i,2} \cdots ,\gamma^L_{i,\ND}\}$, where $\gamma^L_{i,j} \sim Ber\left(p^L_{t,j}\right)$; $1 \leq i \leq K$.
\STATE 2. Calculate the MSE values $\sigma^2_*\left(k\right), k= \left\{1,\ldots, K\right\}$, which would be obtained by activating the corresponding sensors to each of the $K$ samples, according to (\ref{s_blue_algorithm_MSE}).
\STATE 3. Evaluate for each of the $K$ samples the performance metric
\begin{align*}
U\left(k\right)=
\begin{cases}
-\left(w_h \left|\mathcal{S}^H\right|+w_l \left|\mathcal{S}^L\right|\right),&\sigma^2_*\left(k\right) <\epsilon\\
\infty,&\text{Otherwise}
\end{cases}
\end{align*}
where
\begin{align*}
\mathcal{S}^H_j =
\begin{cases}
1,&  \gamma^H_{j,1}=1\\
0,&\text{Otherwise}
\end{cases}
\;\;\;\;\;\text{and}\;\;\;\;\;
\mathcal{S}^L_j =
\begin{cases}
1,&  \gamma^L_{j,1}=1\\
0,&\text{Otherwise}
\end{cases}
\end{align*}
\STATE 4. Calculate the $\beta_t=(1-\rho)$ quantile level of $U_{1:K}$.
\STATE 5.  Update $\mathbf{p}^{H}$ as follows:
        \begin{align*}
            p^{H}_{t,j} = \alpha
						\frac{
						\sum\limits_{i=1}^{K}
						\mathds{1}\left(\mathbf{\Gamma^H_{i,j}}=1\right)
						\overbrace{\mathds{1}\left(U\left(k\right) \geq \beta_t \right)}
						^{\text{Choose elite samples}}
						}
						{\sum\limits_{i=1}^{K}
						\mathds{1}\left(U\left(k\right) \geq \beta_t \right)
						}
            + (1-\alpha)p^H_{t-1,j},
        \end{align*}
\STATE 6.  Update $\mathbf{p}^{L}$ as follows:
        \begin{align*}
            p^{L}_{t,j} = \alpha
						\frac{
						\sum\limits_{i=1}^{K}
						\mathds{1}\left(\mathbf{\Gamma^L_{i,j}}=1\right)
						\overbrace{
						\mathds{1}\left(U\left(k\right) \geq \beta_t \right)}
							^{\text{Choose elite samples}}
						}
						{
						\sum\limits_{i=1}^{K}
						\mathds{1}\left(U\left(k\right) \geq \beta_t \right)
						}
             + (1-\alpha)p^L_{t-1,j},
        \end{align*}
    \ENDWHILE
\STATE 7. For each element in $\mathbf{p}^{H}$ and $\mathbf{p}^{L}$ make the final binary activation decision as follows:
\begin{align*}
\mathcal{S}^H_j =
\begin{cases}
1,&  p_{t,j}^{H} \geq \Psi\\
0,&\text{Otherwise}
\end{cases}
\;\;\;\;\;and\;\;\;\;\;
\mathcal{S}^L_j =
\begin{cases}
1,&  p_{t,j}^{L} \geq \Psi\\
0,&\text{Otherwise}
\end{cases}
\end{align*}
where $\Psi$ is a pre-defined threshold.		
\STATE 8. Evaluate the objective function in  (\ref{SSO}) without the sensors which do not have enough energy. If the QoS constraint is met, then no further steps are required; If the QoS constraint is not met, solve the optimisation problem again, excluding those sensors which were not able to transmit.
\end{algorithmic}
\end{algorithm}

Our system model aims at adapting the selection of the sensors according to both the requirements from the user (location of sensing and statistical accuracy required); as well as the balance between activating high-quality and expensive sensors (eg. weather stations) and low-quality cheap sensors. By ``online" sensor selection we mean that the system selects the ``best" configuration of sensors to activate as a response to a user's query which takes place in a real-time fashion. This differentiates our problem from the so called ``off-line" problem, where the sensors are chosen once and do not change their operation in a responsive manner to user's queries.

\section{Simulations}
In this section, we present extensive simulations to evaluate the performance of the system.
First, in Section \ref{sec:sim_modelfit} we present the accuracy of the field reconstruction using our proposed S-BLUE algorithm for synthetic data. Then in Section \ref{sec:real_modelfit} we present results for the field reconstruction of real data set in the form of wind storm. Finally, in Section \ref{sec:sim_ce} we present the effectiveness of using Cross Entropy based algorithm for sensor selection and activation.
\subsection{Field Reconstruction of Synthetic Data }\label{sec:sim_modelfit}
To generate synthetic data, we used a Squared Exponential kernel
$\C_f\left(\x_1,\x_2;\PSI:=\left\{\sigma^2,l\right\}\right) = \sigma^2\exp\left(\frac{\left\|\x_1-\x_2\right\|^2}{2 l}\right)$ for both spatial random fields, $f\left(\cdot\right)$ and $g\left(\cdot\right)$.
The hyper-parameters for the random spatial phenomenon $f\left(\cdot\right)$ are:
$\PSI_f = \left\{10,1\right\}$ and for the energy harvesting field $\PSI_g = \left\{0.3,1\right\}$.
The additive noise standard deviation is $\sigma_{w}=1$, the mean is $\mu_f=8$. In Fig. \ref{fig:toy_sub0}, we present a single realisation of the field intensity.
In this example we deployed  $4$ high quality and $64$ low quality sensors uniformly in the rectangular region. In Fig. \ref{fig:toy_sub1} we present the spatial field reconstruction for various of activation speeds $T=\{8, 10 ,13, 15\}$.
\begin{figure}[b!]
\centering
\includegraphics[height=4cm]{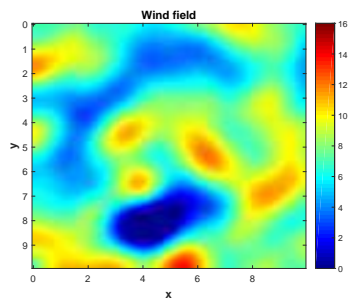}
\caption{Wind field intensity}			
\label{fig:toy_sub0}
\end{figure}

\begin{figure}
\centering
\includegraphics[height=7cm]{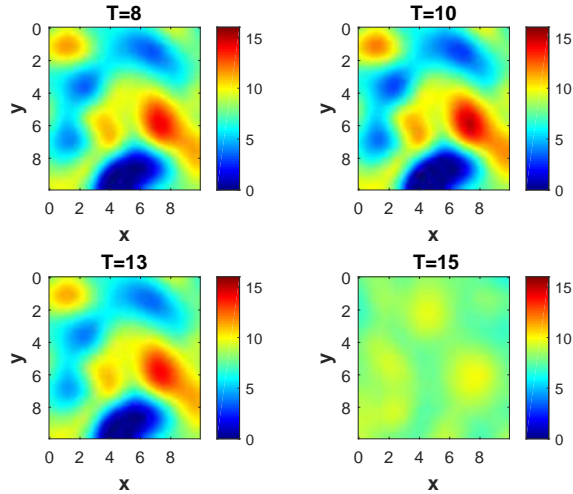}
\caption{S-BLUE wind field reconstruction of Corollary \ref{s_blue_algorithm} for different activation speed thresholds $T=\left\{0,2,5,7\right\}$}				
\label{fig:toy_sub1}
\end{figure}

 In Fig. \ref{fig:toy_sub2} we present the point-wise Root Squared Error (RSE) for these $T$ values.
 Fig. \ref{fig:toy_sub2} illustrates that the estimated wind field closely matches the true wind field for small values of $T$ and does not match the true wind field when $T$ values are high. We also observe that RSE decreases very fast with respect to increasing values of $T$. In addition, we observe that the RSE values are low in the region where many high and low quality sensors are distributed and high in the region where few high and low quality sensors are distributed.

\begin{figure}
\centering
\includegraphics[height=7cm]{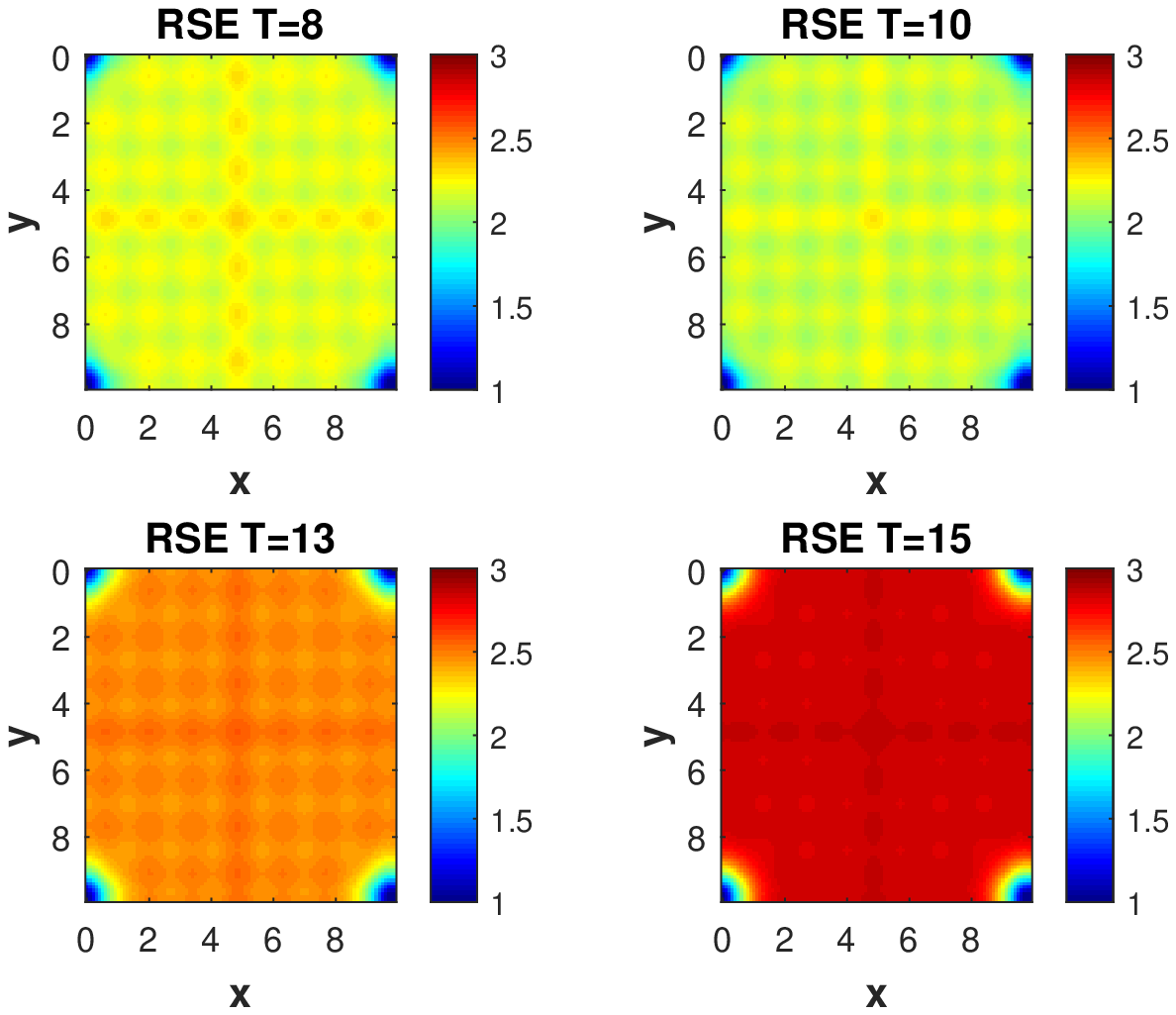}
\caption{Root Mean Squared Error (RMSE) estimation of the wind field intensity}				
\label{fig:toy_sub2}
\end{figure}

In Fig. \ref{fig:mse_sum} we present a quantitative comparison of RSE over $100$ realizations from the spatial field with respect of different number of high quality and low quality sensors when $T=8$, as a function of the number of low quality sensors. We set the number of high quality sensors to $\left\{4, 9, 16, 25\right\}$ and vary the number of low quality sensors from $4$ to $250$. The figure shows how adding low quality sensors aids in reducing the overall RSE.

\begin{figure}[t]
\centering
\includegraphics[height=5cm]{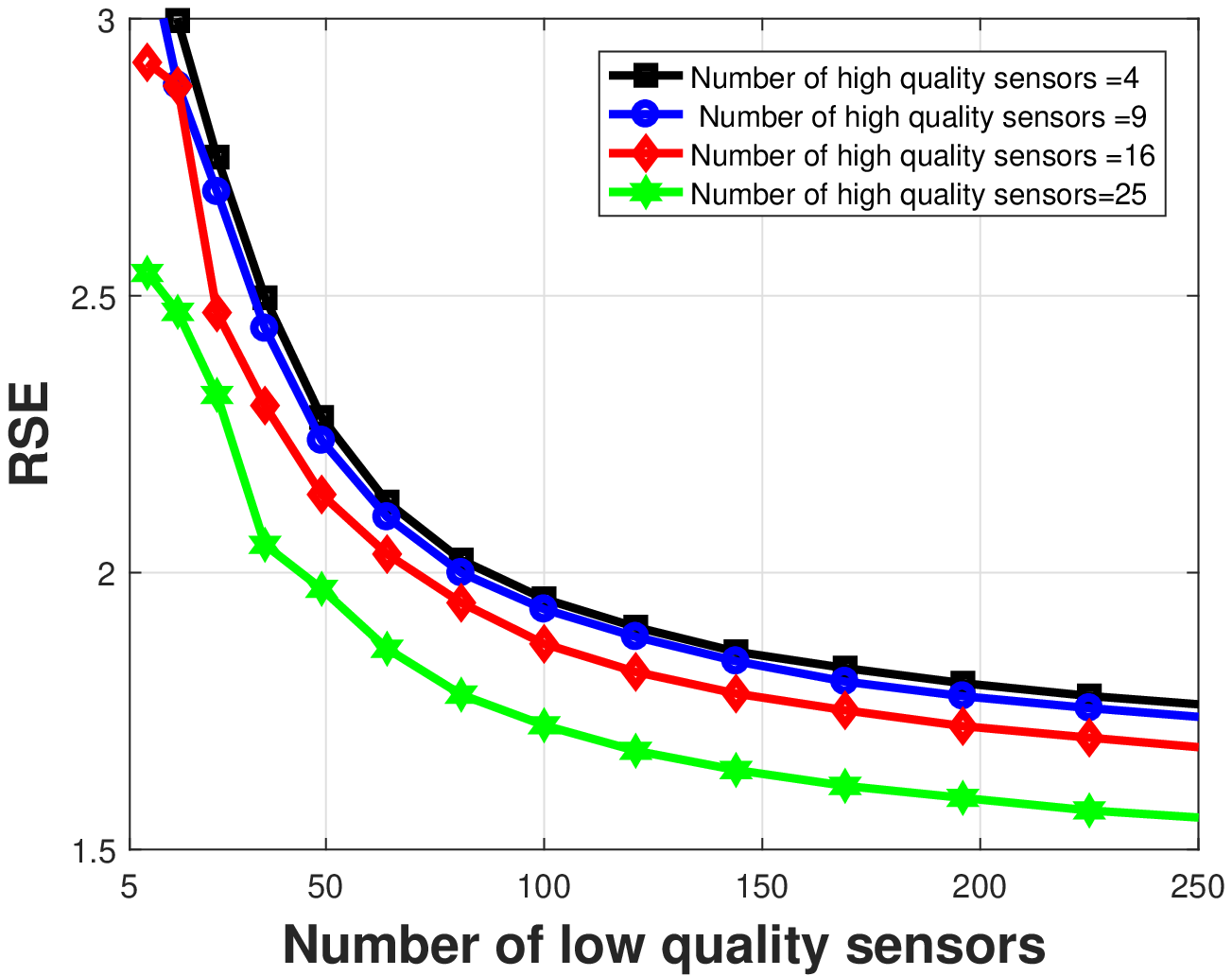}
\caption{ RSE as a function different configurations of number of high and low quality sensors}				
\label{fig:mse_sum}
\end{figure}
\subsection{Field Reconstruction of Storm Surge Data Set} \label{sec:real_modelfit}
In order to test our algorithm on real data sets, we use a publicly available insurance storm surge database known as the Extreme Wind Storms Catalogue \footnote{http://www.met.reading.ac.uk/~extws/database/dataDesc}. The data is available for research as the XWS Datasets: (c) Copyright Met Office, University of Reading and University of Exeter. Licensed under Creative Commons CC BY $4.0$ International License.
This database is comprised of $23$ storms which caused high insurance losses known as `insurance storms' and $27$ storms which were selected because they are the top `non-insurance' storms as ranked by the storm severity index, see details on the website.
The data provided is comprehensive and provides features such as the footprint of the observations on a location grid with a rotated pole at longitude = $177.5$ degrees, latitude = $37.5$ degrees. As discussed in the data description provided with the data-set, this is a standard technique used to ensure that the spacing in km between grid points remains relatively consistent. The footprints are on a regular grid in the rotated coordinate system, with horizontal grid spacing $0.22$ degrees. The data for each of the storms provides a list of grid number and maximum $3$-second gust speed in meters per second. The true locations (longitude and latitude) of the grid points are given in grid locations file.
We selected one storm to analyse, known as Dagmar which took place on 26/12/2011 and affected Finland and Norway.
To calibrate the model we first fit the hyperparameters of the model via Maximum Likelihood Estimation (MLE) procedure.
We used a $2$-D radial basis function, of the following form
\begin{align*}
\C\left(\x_i,\x_j;\PSI\right):=\sigma^2_x\exp{\left(-\frac{\left|x_i-x_j\right|}{l_x}\right)}
 \exp{\left(-\frac{\left|y_i-y_j\right|}{l_y}\right)},
\end{align*}	
thus decomposing the kernel into orthogonal coordinates which we found provided a much more accurate fit. The reason for this is it allows for inhomogeneity through differences in spatial dependence in vertical and horizontal directions, which is highly likely to occur in the types of wind speed data studied. The MLE of the length and scale parameters obtained are given by  $\sigma^2_x=0.1, l_x=0.5$ and $\sigma^2_y=10 ,l_y=0.1$. 
Details on how to estimate the GP hyperparameters can be found in [Chapter 5]\cite{rasmussen2005gaussian}. In our model, historical data is used in order to estimate the hyperparameters of the model at the current time. Then, using these parameters we perform all the inferential tasks.

The left panel of Fig. \ref{fig:google_true_storm} shows the region of interest on the map. Both high and low quality sensors are selected randomly within the region. In this experiment we uniformly deployed $50$ high quality and $250$ low quality sensors. The right panel of Fig. \ref{fig:google_true_storm} shows  the Dagmar storm wind speed intensity. The left column of Fig. \ref{fig:storm_field_varyT} presents the estimated wind speed intensity with varying activation speed thresholds $T$ and the right column presents the spatial RSE values. The figure shows that the true wind intensity field can be recovered for low activation speed, but as the threshold increases, the performance deteriorates.  The RSE is lower at the points where sensors are deployed, and grows with the increase of activation speed threshold. The standard deviation is very low in the middle region, close to a value of $2$ and a bit high in the boundary region where fewer sensors are deployed.

\begin{figure*}
\centering
\begin{subfigure}{.45\textwidth}
  \centering
  \includegraphics[width=.45\linewidth, height=3.5cm]{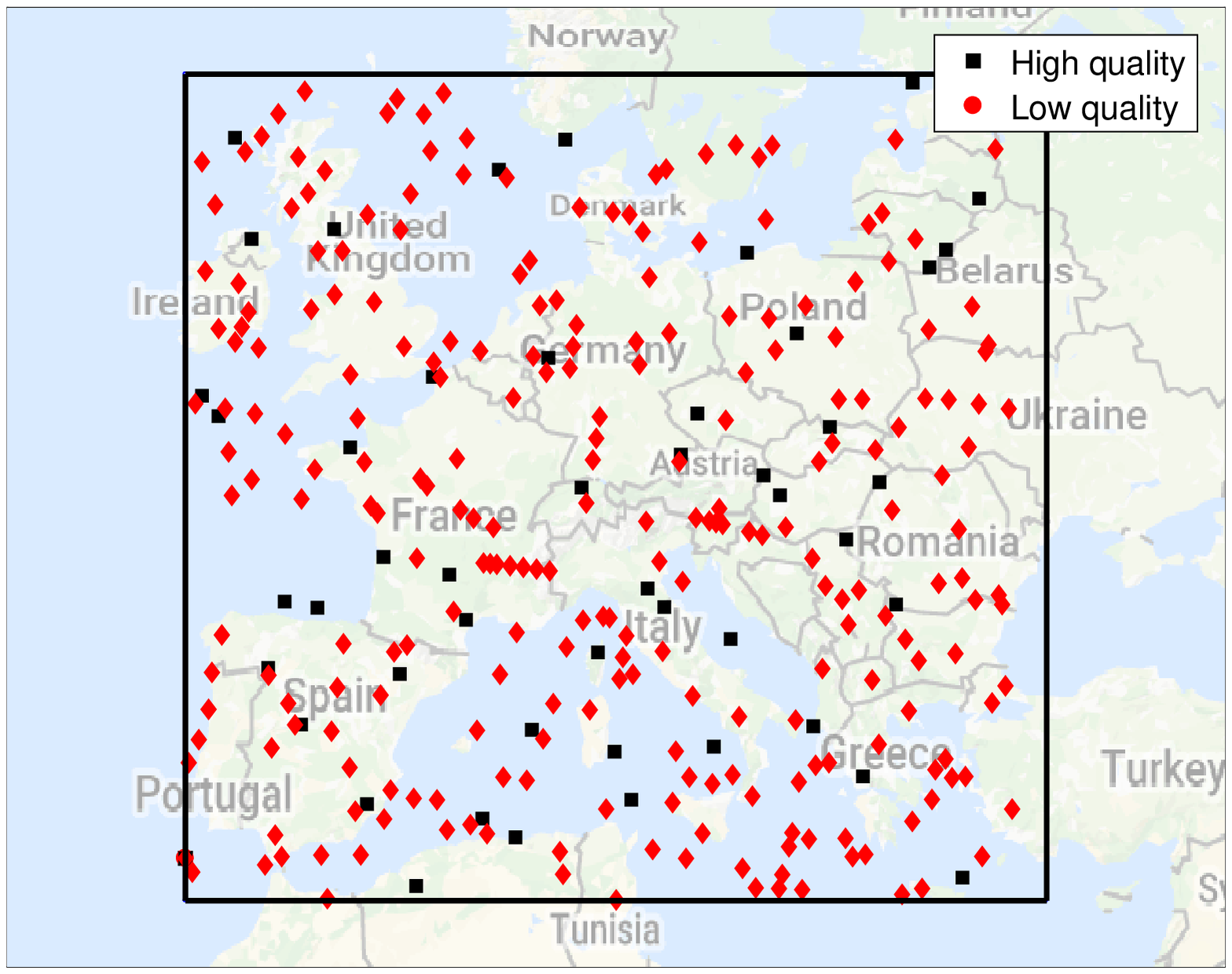}
\end{subfigure}%
\begin{subfigure}{.45\textwidth}
  \centering
  \includegraphics[width=.45\linewidth, height=3.5cm]{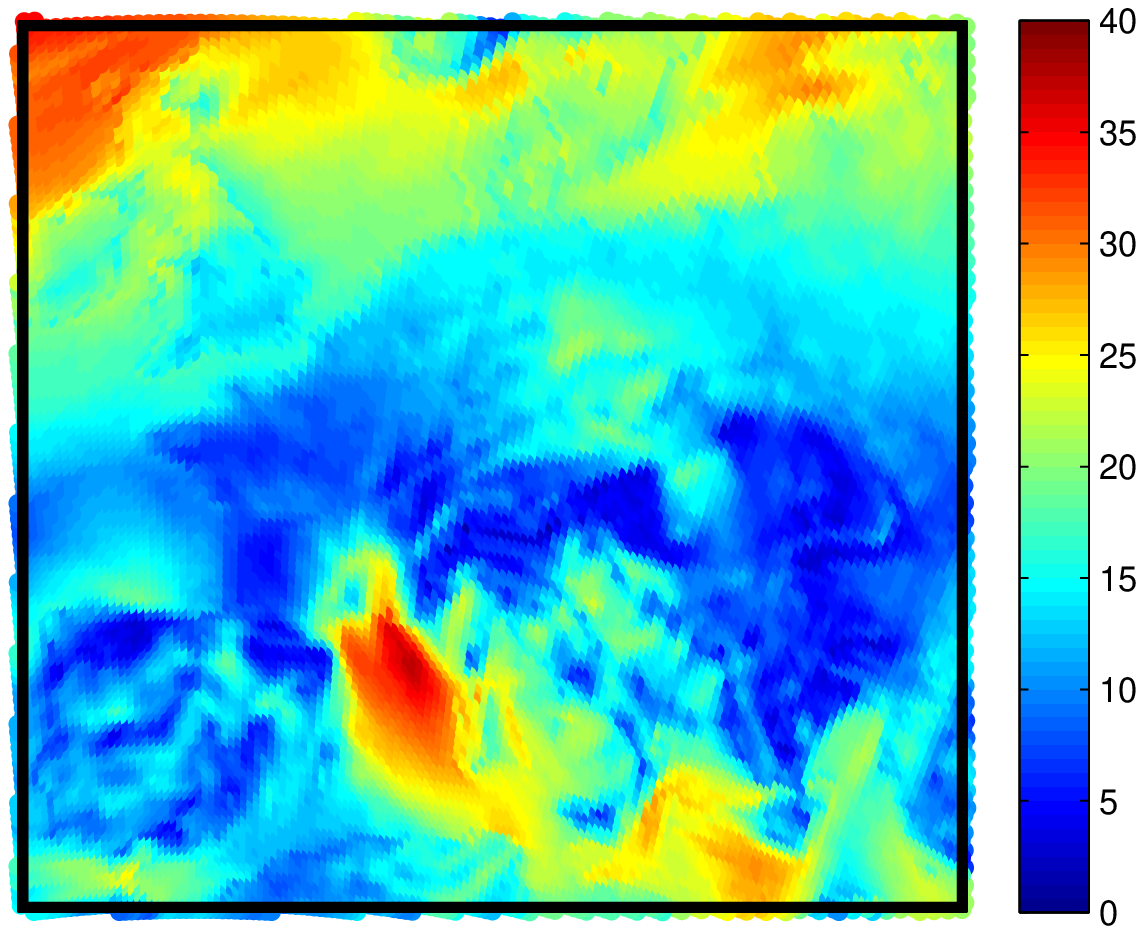}
\end{subfigure}
\caption{Left panel: map of region of interest with sensors locations. Right panel: Dagmar surge storm intensity map}				
\label{fig:google_true_storm}
\end{figure*}

\begin{figure}
\centering
\includegraphics[height=7cm]{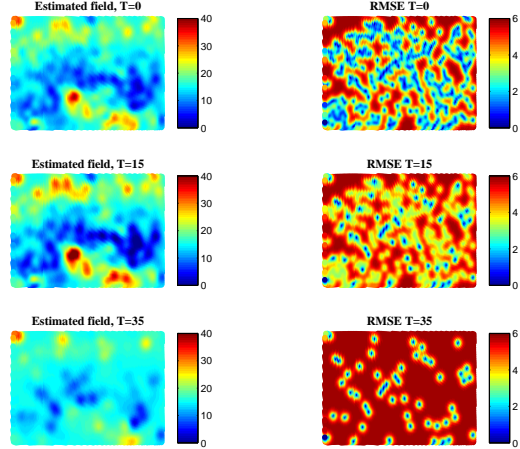}
\caption{True storm field and estimated storm field with various activation speed.}				
\label{fig:storm_field_varyT}
\end{figure}


%
Finally, in Fig. \ref{fig:true_storm_MSE} we present a quantitative comparison of the RMSE for various values of high and low quality sensors. The result shows a clear trend of RMSE with the increasing of high and low quality sensors.

\begin{figure}
\centering
\includegraphics[height=4cm]{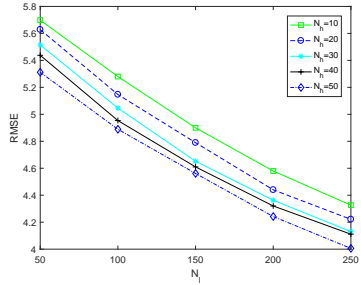}
\caption{RMSE with effect of different number of high and low quality sensors.}				
\label{fig:true_storm_MSE}
\end{figure}
\subsection{Sensor Selection}\label{sec:sim_ce}
In this section we illustrate how our sensor selection algorithm performs. For comparison, we use a optimal selection method which only selects the sensor set collections that minimize the U values and ensures that the QoS criterion is being met.
 The simulation parameters we have are: $\{N_h=5, N_l=10, T=8, w_h=150, w_l = 30, \sigma_w=1, \sigma_g=0.3, k_f(x_*,x_*)=10, x_*=3.5, y_*=3.1, \epsilon=\{5.4, 5.6, 5.8, 6, 6.2\}\}$.

We fix the $N_h=5, N_l=10$. The comparison is shown in Fig. \ref{fig:CEvsOpt}. We change the $\epsilon$ within $\{5.4, 5.6, 5.8, 6, 6.2\}$. We also increase the number of iterations in CE method from 1 to 10. It shows CE method converges quickly to the optimal selection algorithm within 10 iterations for all the $\epsilon$ values.




\begin{figure}[t]
\centering
\includegraphics[height=5cm]{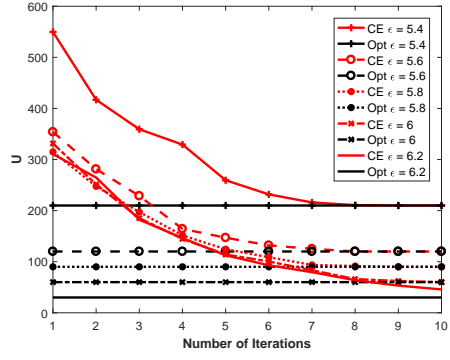}
\caption{Comparison of U values between optimal scheme and CE method with effect of number of iterations.}				
\label{fig:CEvsOpt}
\end{figure}

To compare our method with convex optimization approach, we followed a similar line of thought which was presented in \cite{krause2008near}. We compared the performance of our CEM algorithm with the relaxation-based optimization algorithm. The result shows that for different QoS, $\epsilon$, our CEM has a significant lower cost compared to the convex optimization scenario, as shown in the figure  \ref{fig:vscon}.

\begin{figure}[t]
\centering
\includegraphics[height=5cm]{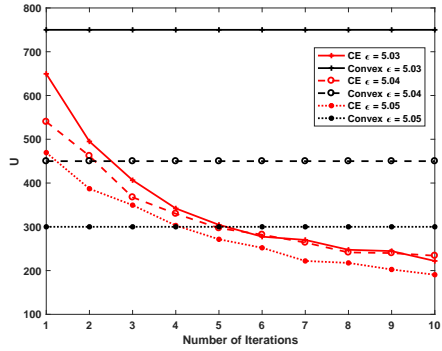}
\caption{Comparison of objective function values between Convex optimization scheme and CE method.}		
\label{fig:vscon}		
\end{figure}

%
%

\section{Conclusions}
We addressed the problem of spatial field reconstruction and sensor selection in heterogeneous sensor networks, containing two types of sensors: expensive, high quality sensors; and cheap, low quality sensors which are activated only if the intensity of the spatial field exceeds a pre-defined activation threshold. In addition, these sensors are powered by means of energy harvesting which impacts their accuracy.
We then addressed the problems of performing \textit{spatial field reconstruction} and \textit{query based sensor set selection with performance guarantee}.
We solved the first problem by developing a low complexity algorithm based on the \textit{spatial best linear unbiased estimator} (S-BLUE).
Next, building on the S-BLUE, developed an efficient algorithm for \textit{query based sensor set selection with performance guarantee}, based on the Cross Entropy method which solves the combinatorial optimization problem in an efficient manner.
We presented a comprehensive study of the performance gain that can be obtained by augmenting
the high-quality sensors with low-quality sensors using both synthetic and real insurance storm surge database known as the Extreme Wind Storms Catalogue.
\appendices
\section{Proof of Lemma \ref{lemma_cross_correlation_1}} \label{lemma_cross_correlation_1_proof}
Using the the law of total expectation, the properties of the GP and the fact that $f\left(\x_k\right)$ and $W_k$ are independent, we obtain that:
\begin{align}
\begin{split}
&\exE_{f_*, Y_H\left(\x_k\right)} \left[f_* Y_H\left(\x_k\right)\right] =
\exE_{f_*, Y_H\left(\x_k\right)}\left[f_* \left(f\left(\x_k\right)+W\left(\x_k\right)\right)\right]\\
&=\exE_{f_*, f\left(\x_k\right)}\left[f_* f\left(\x_k\right)\right]+
\cancel{\exE_{f_*, f\left(\x_k\right)}\left[f_* W\left(\x_k\right)\right]}\\
&=k_f\left(\x_*,\x_k\right).
\end{split}
\end{align}
\section{Proof of Lemma \ref{lemma_cross_correlation_2}} \label{lemma_cross_correlation_2_proof}
\footnotesize
\begin{align*}
\begin{split}
&\exE_{f_*,Y_L\left(\x_k\right)} \left[f_* \; Y_L\left(\x_k\right)\right] \\&=
\exE_{f\left(\x_k\right),\sigV\left(\x_k\right)}
\left[\exE_{f_*,Y\left(\x_k\right)} \left[  f_* Y\left(\x_k\right)|f\left(\x_k\right),\sigV\left(\x_k\right)\right]\right]\\
&=\frac{\mathcal{C}_f\left(\x_*,\x_k\right)}{\mathcal{C}_f(\x_k,\x_k)}\exE_{f\left(\x_k\right),\sigV\left(\x_k\right)}
\left[
f^2\left(\x_k\right)  \mathds{1}\left(f\left(\x_k\right) > T\right)
\right]\\
&= \frac{\mathcal{C}_f\left(\x_*,\x_k\right)}{\mathcal{C}_f(\x_k,\x_k)}
\int\limits_{0}^{\infty}\int\limits_{T}^{\infty}
f^2\left(\x_k\right)
\mathcal{N}\left(f\left(\x_k\right);0, \mathcal{C}_f\left(\x_k,\x_k\right) \right)
\d f\left(\x_k\right)p(\sigV\left(\x_k\right))\\&\d \sigV\left(\x_k\right)\\
\end{split}
\end{align*}
\normalsize
We can derive
\footnotesize
\begin{align*}
\begin{split}
&\exE_{f_*,Y_L\left(\x_k\right)} \left[f_* \; Y_L\left(\x_k\right)\right] \\
&= \frac{\mathcal{C}_f\left(\x_*,\x_k\right)}{\mathcal{C}_f(\x_k,\x_k)}
\frac{1}{\sqrt{\mathcal{C}_f\left(\x_k,\x_k\right)}}\left(\sqrt{\mathcal{C}_f\left(\x_k,\x_k\right)}\right)^3\\&\times\Biggl(\Phi\left(\frac{f\left(\x_k\right)}{\sqrt{\mathcal{C}_f\left(\x_k,\x_k\right)}}\right)
-\frac{f\left(\x_k\right)}{\sqrt{\mathcal{C}_f(\x_k,\x_k)}}\phi\left(\frac{f\left(\x_k\right)}{\sqrt{\mathcal{C}_f\left(\x_k,\x_k\right)}}\right)\Biggr)
\Big|_{ f\left(\x_k\right) =T}^{ f\left(\x_k\right) = \infty}\\
&=\mathcal{C}_f\left(\x_*,\x_k\right)\Biggl(1-\Phi\left(\frac{T}{\sqrt{\mathcal{C}_f\left(\x_k,\x_k\right)}}\right)\\&+\frac{T}{\sqrt{\mathcal{C}_f\left(\x_k,\x_k\right)}}\phi\left(\frac{T}{\sqrt{\mathcal{C}_f\left(\x_k,\x_k\right)}}\right)\Biggr).
\end{split}
\end{align*}
\normalsize
\section{Proof of Lemma \ref{lemma_cross_correlation_3}} \label{lemma_cross_correlation_3_proof}
\footnotesize
\begin{align*}
\begin{split}
&\exE_{Y_H\left(\x_k\right), Y_H\left(\x_j\right)} \left[Y_H\left(\x_k\right) \;Y_H\left(\x_j\right)\right] \\&=
\exE_{f\left(\x_k\right), f\left(\x_j\right),W\left(\x_k\right),W\left(\x_j\right)} \left[\left(f\left(\x_k\right)+W\left(\x_k\right)\right) \left(f\left(\x_j\right)+W\left(\x_j\right)\right)\right]\\
 &=
\exE_{f\left(\x_k\right), f\left(\x_j\right)} \left[f\left(\x_k\right)f\left(\x_j\right)\right]+
\exE_{W\left(\x_k\right),W\left(\x_j\right)}\left[W\left(\x_k\right)\; W\left(\x_j\right)\right]\\
&= k_f\left(\x_k,\x_j\right)+ \mathds{1}\left(k=j\right)\sigW.
\end{split}
\end{align*}
\normalsize
\section{Proof of Lemma \ref{lemma_cross_correlation_4}} \label{lemma_cross_correlation_4_proof}
\footnotesize
\begin{align*}
\begin{split}
&\exE_{Y_H\left(\x_k\right), Y_L\left(\x_j\right)} \left[Y_H\left(\x_k\right)\; Y_L\left(\x_j\right)\right]\\& =
\exE_{f\left(\x_j\right),\sigV\left(\x_j\right)}
\left[
\exE_{Y_H\left(\x_k\right), Y_L\left(\x_j\right)}
\left[
\left(Y_H\left(\x_k\right) Y_L\left(\x_j\right)|f\left(\x_j\right),\sigV\left(\x_j\right)\right)
\right]
\right]  \\
&=
\exE_{f\left(\x_j\right),\sigV\left(\x_j\right)}
\Bigl[
\exE_{f\left(\x_k\right),W\left(\x_k\right),V\left(\x_j\right)}
\Biggl[
\Bigl(\left(f\left(\x_k\right)+W\left(\x_k\right)\right) \left(f\left(\x_j\right)+V\left(\x_j\right)\right)\\&|f\left(\x_j\right),\sigV\left(\x_j\right)\Bigr)
\Biggr]
\Bigr]\\
&=\exE_{f\left(\x_j\right),\sigV\left(\x_j\right)}\Biggl[\frac{\mathcal{C}_f\left(\x_k,\x_j\right)}{\mathcal{C}_f(\x_j,\x_j)} f\left(\x_j\right)^2  \mathds{1}\left(f\left(\x_j\right) > T\right)\Biggr]\\
\end{split}
\end{align*}
\normalsize
Now we can follow the derivation in Appendix \ref{lemma_cross_correlation_2_proof}, and we can get
\footnotesize
\begin{align*}
\begin{split}
&\exE_{Y_H\left(\x_k\right), Y_L\left(\x_j\right)} \left[Y_H\left(\x_k\right)\; Y_L\left(\x_j\right)\right] \\ &=\mathcal{C}_f\left(\x_k,\x_j\right)\Biggl(1-\Phi\left(\frac{T}{\sqrt{\mathcal{C}_f\left(\x_j,\x_j\right)}}\right)\\&+\left(\frac{T}{\sqrt{\mathcal{C}_f\left(\x_j,\x_j\right)}}\right)\phi\left(\frac{T}{\sqrt{\mathcal{C}_f\left(\x_j,\x_j\right)}}\right)\Biggr).
\end{split}
\end{align*}
\normalsize
\section{Proof of Lemma \ref{cross_correlation_6}} \label{theorem_cross_correlation_6_proof}
\footnotesize
\begin{align*}
\begin{split}
&\exE_{Y_L\left(\x_k\right)} \left[Y_L\left(\x_k\right) \;Y_L\left(\x_k\right)\right]\\&=
\exE_{Y_L\left(\x_k\right),V\left(\x_k\right)} \left[\left(f\left(\x_k\right)+V\left(\x_k\right)\right)^2\mathds{1}\left(f\left(\x_k\right)>T\right)\right]\\
&=
\exE_{f\left(\x_k\right)}
\left[
\exE_{V\left(\x_k\right)} \left[
\left(f\left(\x_k\right)+V\left(\x_k\right)\right)^2|f\left(\x_k\right)
\right]\mathds{1}\left(f\left(\x_k\right)>T\right)
\right]\\
&=
\exE_{f\left(\x_k\right)}
\Biggl[
\exE_{V\left(\x_k\right)}
\Bigl[
\left(f\left(\x_k\right)^2+2 f\left(\x_k\right) V\left(\x_k\right)+V\left(\x_k\right)^2|f\left(\x_k\right)\right)
\\&\times\mathds{1}\left(f\left(\x_k\right)>T\right)\Bigr]\Biggr]\\
&=\exE_{f\left(\x_k\right)}\left[f\left(\x_k\right)^2\mathds{1}\left(f\left(\x_k\right)>T\right)\right]
\\&+
\exE_{V\left(\x_k\right)}
\left[V\left(\x_k\right)^2\right]\\
&=\mathcal{C}_f\left(\x_k,\x_k\right)\Biggl(1-\Phi\left(\frac{T}{\sqrt{\mathcal{C}_f\left(\x_k,\x_k\right)}}\right)\\&+\left(\frac{T}{\sqrt{\mathcal{C}_f\left(\x_k,\x_k\right)}}\right)\phi\left(\frac{T}{\sqrt{\mathcal{C}_f\left(\x_k,\x_k\right)}}\right)\Biggr)+\exp\left(\mu_g\left(\x_k\right) +
 \frac{\mathcal{C}_g\left(\x_k,x_k\right)}{2}\right).
\end{split}
\end{align*}
\normalsize
\section{proof of lemma \ref{theorem_cross_correlation_7}}\label{thereom_cross_correlation_7_proof}
\footnotesize
\begin{align}
\begin{split}
&\left[\Qfour\right]_{k,j}:=
\exE_{Y_L\left(\x_k\right),Y_L\left(\x_j\right)} \left[Y_L\left(\x_k\right)\; Y_L\left(\x_j\right)\right]\\
&=\exE_{f\left(\x_k\right),f\left(\x_j\right),\sigV\left(\x_k\right),\sigV\left(\x_j\right)}
\Big[
\left(f\left(\x_k\right)+V\left(\x_k\right)\right)\left(f\left(\x_j\right)+V\left(\x_j\right)\right)
\\&\times\mathds{1}\left(f\left(\x_k\right)>T,f\left(\x_j\right)>T\right)
|\sigV\left(\x_k\right),\sigV\left(\x_j\right)\Big]\\
\end{split}
\end{align}
\normalsize
Note in the above equations, all the cross moments terms relating to the product of $V\left(\x_k\right)V\left(\x_j\right)$, $f\left(\x_k\right)V\left(\x_j\right)$ and $f\left(\x_j\right)V\left(\x_j\right)$ will become zero since there is independence between these terms. So the equation reduces to:
\footnotesize
\begin{align}
\begin{split}
&\left[\Qfour\right]_{k,j} =\int\limits_T^{\infty}
\int\limits_T^{\infty} p\left(f\left(\x_k\right), f\left(\x_j\right)\right)
f\left(\x_k\right) f\left(\x_j\right) \d f\left(\x_k\right)\; \d f\left(\x_j\right)\\
&=\exE_{f\left(\x_k\right),f\left(\x_j\right)}
\Big[
f\left(\x_k\right)f\left(\x_j\right)
\mathds{1}\left(f\left(\x_k\right)>T,f\left(\x_j\right)>T\right)
\Big]\\
&=\exE_{f_s\left(\x_k\right)\sqrt{\mathcal{C}_f\left(\x_k,\x_k\right)},f_s\left(\x_j\right)\sqrt{\mathcal{C}_f\left(\x_j,\x_j\right)}}\\
&\Big[
\left(f_s\left(\x_k\right)\sqrt{\mathcal{C}_f\left(\x_k,\x_k\right)}\right)\left(f_s\left(\x_j\right)\sqrt{\mathcal{C}_f\left(\x_j,\x_j\right)}\right)
\\&\times\mathds{1}\left(f_s\left(\x_k\right)>T_k,f_s\left(\x_j\right)>T_j\right)
\Big]\\
&=\sqrt{\mathcal{C}_f\left(\x_k,\x_k\right)\mathcal{C}_f\left(\x_j,\x_j\right)}\int\limits_{T_k}^{\infty}
\int\limits_{T_j}^{\infty}p\left(f_s\left(\x_k\right), f_s\left(\x_j\right)\right)f_s\left(\x_k\right)f_s\left(\x_j\right)\\&\d f_s\left(\x_k\right)\d f_s\left(\x_j\right)\\
&=\sqrt{\mathcal{C}_f\left(\x_k,\x_k\right)\mathcal{C}_f\left(\x_j,\x_j\right)}\exE_{f_s\left(\x_k\right),f_s\left(\x_j\right)}\Big[f_s\left(\x_k\right)f_s\left(\x_j\right)\Big].
\end{split}
\end{align}
\normalsize
Finally, utilising Theorem \ref{theorem_hermite} we obtain the result.
\section{Proof of Lemma \ref{lemma_mean_y}} \label{lemma_mean_y_proof}
\footnotesize
\begin{equation*}
\begin{split}
&\exE[Y\left(\x_k\right)]=\exE_{\sigV}\left[\exE\left[Y(\x_k)|\sigV\right]\right]\\&=\exE_{\sigV}\left[\exE\left[(f(\x_k)+V(\x_k))\mathds{1}(f(\x_k)\geq T)|\sigV+V(\x_k)\mathds{1}(f(\x_k)< T)|\sigV\right]\right]\\
&=\exE_{\sigV}\left[\exE\left[f(\x_k)\mathds{1}(f(\x_k)\geq T)\right]\right]\\
&=\exE_{\sigV}\left[\int f(\x_k) p(f(\x_k))\mathds{1}(f(\x_k)\geq T)d f(\x_k)\right]\\
&=\exE_{\sigV}\left[\frac{1}{\sqrt{\mathcal{C}_f\left(\x_k,\x_k\right)}}\int_T^{+\infty}f(\x_k)\phi\left(\frac{f(\x_k)}{\sqrt{\mathcal{C}_f\left(\x_k,\x_k\right)}}\right)d f(\x_k)\right]\\
&=-\sqrt{\mathcal{C}_f\left(\x_k,\x_k\right)}\left(\phi\left(\frac{f(\x_k)}{\sqrt{\mathcal{C}_f\left(\x_k,\x_k\right)}}\right)\right)\big|_T^\infty\\
&=\sqrt{\mathcal{C}_f\left(\x_k,\x_k\right)}\phi\left(T_k\right).
\end{split}
\end{equation*}

\section*{Acknowledgment}
This work was supported by the Korea Institute of Energy Technology Evaluation and Planning (KETEP) and the Ministry of Trade, Industry \& Energy (MOTIE) of the Republic of Korea (No. 20148510011150).
\bibliographystyle{IEEEtran}
\bibliography{references}
\end{document}